\documentclass[11pt]{article}
\usepackage{pdfsync}
\usepackage{amsmath}
\usepackage{amssymb}
\usepackage{amscd}
\usepackage{mathrsfs}
\usepackage{color}
\usepackage{ulem}
\usepackage{bm}
\usepackage{graphicx}
\newtheorem{theorem}{Theorem}[section]
\newtheorem{prop}[theorem]{Proposition}
\newtheorem{cor}[theorem]{Corollary}
\newtheorem{lemma}[theorem]{Lemma}
\newtheorem{algorithm}[theorem]{Algorithm}
\newtheorem{remark}[theorem]{Remark}

\newtheorem{define}[theorem]{Definition}

\newtheorem{notation}[theorem]{Notation}

\newcommand{\ord}{\mbox{\rm ord}}

\definecolor{highlight}{rgb}{.5,0,.5}

\newcommand{\zero}{{\mbox{\rm Zero}}}

\newcommand{\mat}{{\mbox{\rm Mat}}}

\newcommand{\gal}{{\mbox{\rm Gal}}}

\newcommand{\rank}{\mbox{\rm {rank}}}
\newcommand{\Tr}{\mbox{\rm {Tr}}}

\def\gl{{\mathfrak{g}l}}

\def\dsA{{A^{\oplus n}}}
\def\ssA{{S^{\leq d}A^{\oplus n}}}

\def\ssV{{(V^n)^{\circledS \leq d}}}

\def\diag{\hbox{\rm diag}}
\def\GL{{\rm GL}}
\def\stab{{\rm stab}}
\def\Zero{{\rm Zero}}
\def\num{{\rm num}}
\def\Mat{{\rm Mat}}
\def\sp{{\rm span}}

\def\calH{{\mathcal H}}

\def\bZ {{\mathbb{Z}}}

\def\bN{{\mathbb{N}}}
\def\calN{{\mathcal{N}}}

\def\calF{{\mathcal F}}

\def\calW{{\mathcal W}}
\def\calH{{\mathcal H}}
\def\calG{{\mathcal G}}
\def\calL{{\mathcal L}}

\def\calS{{\cal S}}

\def\bfb{{\mathbf b}}
\def\bfc{{\mathbf c}}
\def\bfe{{\mathbf e}}
\def\bfp{{\mathbf p}}
\def\bfu{{\mathbf u}}
\def\bfv{{\mathbf v}}
\def\bfw{{\mathbf w}}
\def\bfx{{\mathbf x}}

\def\bfm{{\mathbf m}}

\def\bfh{{\mathbf h}}
\def\bfz{{\mathbf z}}

\makeatletter

\newcommand{\Rmnum}[1]{\expandafter\@slowromancap\romannumeral #1@}
\makeatother

\begin{document}

\title{Hrushovski's Algorithm for\\Computing the Galois Group of a Linear Differential Equation }
\author{Ruyong Feng\footnote{ryfeng@amss.ac.cn. This work is partially supported by a National Key Basic Research Project of China (2011CB302400) and
by a grant from NSFC (60821002).} \\ KLMM, AMSS, Chinese Academy of Sciences, \\Beijing 100190, China
}
\date{} \maketitle

\begin{abstract} We present a detailed and simplified version of Hrushovski's algorithm that determines the Galois group of a linear differential equation. There are three major ingredients in this algorithm. The first is to look for a degree bound for proto-Galois groups, which enables one to compute one of them. The second is to determine the identity component of the Galois group that is the pullback of a torus to the proto-Galois group. The third is to recover the Galois group from its identity component and a finite Galois group.     \end{abstract}


\section{Introduction}
\label{sec-preliminary}
In \cite{hrushovski}, Hrushovski developed an algorithm to compute the Galois groups for general linear differential equations. To the best of my knowledge, this is the first algorithm that works for all linear differential equations with rational function coefficients. Before Hrushovski's results, the known algorithms were only valid for linear differential equations of special types, for instance, low order or completely reducible equations. The algorithm due to Kovacic (\cite{kovacic}) deals with the second order equations. In \cite{sing-ulmer}, the authors determined the structural properties of the Galois groups of second and third order linear differential equations. In many cases these properties can be used to determine the Galois groups. In \cite{compoint-singer}, the authors gave an algorithm to compute the Galois group of linear differential equations that are completely reducible. The reader is referred to \cite{singer09, put} for the survey of algorithmic aspects of Galois groups  and the references given there for more results. In particular, in \cite{singer09}, the author gave a clear explanation of the method based on Tannakian philosophy and introduced the various techniques that were used in the known algorithms.

Throughout the paper, $C$ denotes an algebraically closed field of characteristic zero and $k=C(t)$ is the differential field with the usual derivation $\delta=\frac{d}{dt}$. The algebraic closure of $k$ is denoted by $\bar{k}$. Linear differential equations we consider here will be of the matrix form:
\begin{equation}
\label{lodes}
   \delta(Y)=AY,
\end{equation}
where $Y$ is a vector with $n$ unknowns and $A$ is an $n\times n$ matrix with entries in $k$.  Denote the Picard-Vessiot extension field of $k$ for (\ref{lodes}) by $K$ and the solution space of (\ref{lodes}) by $V$. Then the Galois group of (\ref{lodes}) over $k$, denoted by  $\gal(K/k)$, are defined as the group of differential automorphisms of $K$ that keep all elements of $k$ fixed. For brevity, we usually use $\calG$ to denote this group. $\calG$ is a subgroup of $\GL(V)$. A matrix in $\GL_n(K)$ whose columns form a basis of $V$ is called a fundamental matrix of (\ref{lodes}).

Elements of $V^n$ are vectors with $n^2$ coordinates. For the sake of convenience, elements of $V^n$ are also written in the matrix form. In such a case, $V^n=\{Fh | h\in \Mat_n(C)\}$, where $F$ is a fundamental matrix of (\ref{lodes}). Set $$V^n_{inv}=\{Fh|h\in \GL_n(C)\}.$$
Then $V_{inv}^n$ is an open subset of $V^n$. Note that if $F=I_n$, then $V^n_{inv}=\GL_n(C)$. In this paper, we will always use $X$ to denote the $n\times n$ matrix whose entries are indeterminates $x_{i,j}$. Without any possible ambiguity, we will also use $X$ to denote the set of indeterminates. Let $Z$ be a subset of $V^n_{inv}$.
 $Z$ is said to be a Zariski closed subset of $V_{inv}^n$ if there are polynomials $P_1(X),\cdots, P_m(X)$ such that
$Z=\zero(P_1(X),\cdots, P_m(X))\cap V_{inv}^n$. In this case, we also say that $Z$ is defined by $P_1(X), \cdots,P_m(X)$. We will use $N_d(V_{inv}^n)$ to denote the set of all subsets of $V_{inv}^n$, which are defined by finitely many polynomials with degree not greater than $d$. Note that here we have no requirement for the coefficients of these polynomials. Suppose that $\tilde{k}$ is an extension field of $k$ and $Z\subseteq V_{inv}^n$. If $Z$ can be defined by polynomials with coefficients in $\tilde{k}$, then $Z$ is said to be $\tilde{k}$-{\it definable}. The {\it stabilizer} of $Z$, denoted by $\stab(Z)$, is defined as the subgroup of $\GL(V)$ whose elements keep $Z$ unchange. Let $Z\in N_d(V_{inv}^n)$. In case that we emphasize the degree $d$ of defining polynomials, we will also say ``$Z$ is bounded by $d$".

Let $F$ be a fundamental matrix. For any $\sigma\in \GL(V)$, there exists $[\sigma]\in \GL_n(C)$ such that $\sigma(F)=F[\sigma]$. The map
$\phi_F: \GL(V) \rightarrow \GL_n(C)$, given by $\phi_F(\sigma)=[\sigma]$, is a group isomorphism and $\phi_F(\calG)$ is an algebraic subgroup of $\GL_n(C)$. Let $\calH$ be a subgroup of $\GL_n(V)$. For ease of notations, we use $\calH_F$ to denote $\phi_F(\calH)$. $\calH$ is said to be an algebraic subgroup if so is $\calH_F$. Assume that $\calH$ is an algebraic subgroup. Then $\calH^\circ, \calH^t$ are used to denote the pre-images of $\calH_F^\circ$ and $\calH_F^t$, where $\calH_F^\circ$ denotes the identity component of $\calH_F$ and $\calH_F^t$ is the intersection of kernels of all characters of $\calH_F$. $\calH$ is said to be bounded by $d$ if so is $\calH_F$.

The key point of Hrushovski's algorithm is that one can compute an integer $\tilde{d}$ such that there is an algebraic subgroup $\calH$ (or $\calH_F$) of $\GL(V)$ bounded by $\tilde{d}$ satisfying
$$
   (*): (\calH^\circ)^t\unlhd \calG^\circ \leq \calG \leq \calH,\,\,\mbox{or}\,\,\,\, (\calH_F^\circ)^t\unlhd \calG_F^\circ \leq \calG_F \leq \calH_F.
$$
For simplicity of presentation, we introduce the following notion.
\begin{define}
\label{define-pregaloisgroup}
The algebraic group $\calH$ $($or $\calH_F$$)$ in $(*)$ is called a {\it proto-Galois group} of (\ref{lodes}).
\end{define}
Roughly speaking, Hrushovski's approach includes the following steps.
\begin{itemize}
  \item [$(S1)$](proto-Galois groups). One can compute an integer $\tilde{d}$ only depending on $n$ such that there is a proto-Galois group of (\ref{lodes}), which is bounded by $\tilde{d}$.
Let $\calH$ be the intersection of the stabilizers of $k$-definable elements of $N_{\tilde{d}}(V^n_{inv})$. Then $\calH$ is a desired proto-Galois group of (\ref{lodes}).
\item [$(S2)$] (The toric part).
  Compute $\calH_F^\circ$ and let $\chi_1,\cdots, \chi_l$ be generators of the character group of $\calH_F^\circ$. Then the map $\varphi=(\chi_1,\cdots,\chi_l)$ is a morphism from $\calH_F^\circ$ to $(C^*)^l$. $\varphi(\calG_F^\circ)$ is the Galois group of some exponential extension $E$ of $\tilde{k}$ over $\tilde{k}$, where $\tilde{k}$ is an algebraic extension of $k$. One can find $E$ by computing the hyperexponential solutions of some symmetric power system of (\ref{lodes}). Pulling $\varphi(\calG_F^\circ)$ back to $\calH_F^\circ$, one gets $\calG_F^\circ$.
\item [$(S3)$](The finite part).
  Find a finite Galois extension $k_G$ of $k$ and a $k_G$-definable subset $Z$ of $V^n_{inv}$ such that $\calG^\circ=\stab(Z)$. Let $Z_1=Z, Z_2, \cdots, Z_m$ be the orbit of $Z$ under the action of
  $\gal(k_G/k)$.
  Then $\calG=\cup_{i=1}^m\{\sigma\in \GL(V)| \sigma(Z)=Z_i \}$.
\end{itemize}
We follow Hrushovski's approach, but take out the logical language and elaborate the details of the proofs that were only sketched in his paper. We hope this will be helpful for the reader to understand Hrushovski's approach. Meanwhile, we simplify the first step of his approach. In spite of calculating all $k$-definable elements, we only need compute one $k$-definable element to obtain a proto-Galois group. Hrushovski showed in the part \Rmnum{3} of \cite{hrushovski} how to compute the integer $\tilde{d}$ as claimed in (S1). As well as providing the detailed explanations of his proofs, we present an explicit estimate of the integer $\tilde{d}$.

The paper is organized as follows. In Appendix A, we describe a method to find a bound for $k$-definable elements of $N_d(V^n_{inv})$. In Appendix B, an explicit estimate of the integer $\tilde{d}$ that bounds the proto-Galois groups is presented. These bounds will guarantee the termination of the algorithm. In Sections 2 and 3, we show how to compute a proto-Galois group and then the Galois group.  Some computation details are omitted in these sections and will be completed in Section 4.

{\bf Acknowledgements}. Special thanks go to Michael F. Singer for his numerous significant suggestions that improve the paper a lot. In preparing the paper, the author was invited by Michael F. Singer to visit North Carolina State University for two weeks. The author thanks him for the invitation and financial support. Many thanks also go to Shaoshi Chen, Ziming Li and Daniel Rettstadt for their valuable discussions.
\section{proto-Galois groups}
\label{sec-pregaloisgroup}
This section will be devoted to finding a proto-Galois group of (\ref{lodes}). Let $d \in \bZ_{\geq 0}\cup \{\infty\}$ and $I \subseteq k[x_{1,1},\cdots, x_{n,n}]$. $I_{\leq d}$ denotes the set of polynomials in $I$ with degree not greater than $d$. Set
\begin{equation}
\label{eqn-partialalgrelations}
   I_{F,d}=\left\{ P(X)\in k[x_{1,1},\cdots, x_{n,n}]_{\leq d}\,\,\left |\,\, P(F)=0\right.\right\} \,\,\mbox{and}\,\,Z_{F,d}=\Zero(I_{F,d})\bigcap V_{inv}^n
\end{equation}
When $d=\infty$, we use $I_F$ and $Z_F$ for short. $I_F$ consists of algebraic relations among entries of the fundamental matrix $F$. The Galois group of (\ref{lodes}) is considered as the subgroup of $\GL(V)$ that preserves $I_F$. Precisely,
$$
  \calG=\{\sigma\in \GL(V)\,\,|\,\, \sigma(Z_F)=Z_F\}.
$$
Hence once $I_F$ is computed, $\calG$ will be determined. In general, it is hard to calculate $I_F$.
While given a nonnegative integer $d$, the results in Appendix A enable us to compute $I_{F,d}$ (see Section \ref{subsec-computpregaloisgroup}). Moreover, we shall show that if $d$ is large enough, then $\stab(Z_{F,d})$ will be a proto-Galois group of (\ref{lodes}). Corollary \ref{cor-boundeddegree} in Appendix B tells us how large the integer $d$ is enough.
Let us start with two lemmas.
 \begin{lemma}
 \label{lem-definable}
   Assume that $U$ is a $C$-definable Zariski closed subset of $\GL_n(C)$ such that $\calG\subseteq \stab(FU)$. Then $FU$ is a $k$-definable Zariski closed subset of $V^n_{inv}$. Moreover if $U$ is bounded by $d$ then so is $FU$.
 \end{lemma}
 \begin{proof}
     Assume that $U$ is bounded by $d$. Let
     $$
         J=\left\{P(X)\in C[x_{1,1},x_{1,2},\cdots,x_{n,n}]_{\leq d}\,\,|\,\,\forall \,\,u\in U, P(u)=0 \right\}.
     $$
      Since $U$ is bounded by $d$,
     $U=\Zero(J) \cap \GL_n(C)$. Let
     $
        \tilde{I}
     $ be the ideal in $K[x_{1,1},x_{1,2},\cdots,x_{n,n}]$ generated by $\{P(F^{-1}X)| P(X)\in J\}$. Set
    $$
      \tilde{I}_{\leq d}=\tilde{I}\cap K[x_{1,1},x_{1,2},\cdots, x_{n,n}]_{\leq d}.
    $$
    Then $\{P(F^{-1}X)| P(X)\in J\}\subseteq \tilde{I}_{\leq d}$ and $FU=\Zero(\tilde{I}_{\leq d})\cap V^n_{inv}$. Moreover, assume that $P(X)$ is a polynomial in $K[x_{1,1},x_{1,2},\cdots,x_{n,n}]_{\leq d}$ satisfying that $P(Fu)=0$ for any $u\in U$. Then one can easily see that $P(FX)$ is a $K$-linear combination of finitely many elements in $J$. Therefore $P(X)\in \tilde{I}_{\leq d}$. Now let
    $$
       I_{\leq d}=\tilde{I}_{\leq d}\cap k[x_{1,1},x_{1,2},\cdots, x_{n,n}].
    $$
    We will show that $I_{\leq d}$ generates $\tilde{I}$. For this, it suffices to prove that $I_{\leq d}$ generates $\tilde{I}_{\leq d}$, since $\tilde{I}_{\leq d}$ generates $\tilde{I}$.
    We will use the similar argument as in (p.23, \cite{put-singer}) to prove this. Use $\langle I_{\leq d}\rangle$ to denote the ideal of $K[x_{1,1},x_{1,2},\cdots,x_{n,n}]$ generated by $I_{\leq d}$.
    Assume that $\tilde{I}_{\leq d}$ is not a subset of $\langle I_{\leq d}\rangle$. Pick $Q(X)\in \tilde{I}_{\leq d}\setminus \langle I_{\leq d}\rangle$ such that $\num(Q(X))$, the number of the monomials of $Q(X)$, is minimal. If $\num(f)=1$, it is clear that $Q(X)\in \langle I_{\leq d}\rangle$. Hence $\num(Q(X))>1$. Without loss of generality, we may assume that one of the coefficients of $Q(X)$ equals one and one of them, say $c$, is not in $k$. Let $\sigma\in \calG$. We use $Q_\sigma(X)$ to denote the image of $Q(X)$ by applying $\sigma$ to the coefficients of $Q(X)$. For every $u\in U$, since $Q(\sigma^{-1}(Fu))=0$,
    $$\sigma(Q(\sigma^{-1}(Fu))=Q_\sigma(\sigma(\sigma^{-1}(Fu)))=Q_\sigma(Fu)=0.$$
    It implies that $Q_\sigma(X)\in \tilde{I}_{\leq d}$ for all $\sigma\in \calG$.
     Then the minimality of $\num(Q(X))$ implies that both $Q(X)-Q_\sigma(X)$ and $c^{-1}Q(X)-\sigma(c)^{-1}Q_\sigma(X)$ are in $\langle I_{\leq d}\rangle$. Therefore
     $$\forall \,\, \sigma\in \calG,\,\,(\sigma(c)^{-1}-c^{-1})Q(X)\in \langle I_{\leq d}\rangle.$$
      Since $c\notin k$, there is $\sigma\in \calG$ such that $\sigma(c)\neq c$. So $Q(X)\in \langle I_{\leq d}\rangle$, a contradiction. Hence $I_{\leq d}$ generates $\tilde{I}_{\leq d}$.
 \end{proof}
 The correctness of the statement below is almost obvious. However, since it will be used frequently, we state it as a lemma.
  \begin{lemma}
     \label{lem-stabilizer}
     Let $H$ be a subgroup of $\GL_n(C)$ and $\calN=\stab(FH)$. Then $\calN_F=H$.
  \end{lemma}
  \begin{proof}
    One can easily see that $\calN_F\subseteq H$. Assume that $h\in H$. Then there is an element $\sigma_h$ in $\GL(V)$ satisfying that
    $\sigma_h(F)=Fh$. Now for any $h'\in H$,
    $$
       \sigma_h(Fh')=Fhh'\in FH \,\,\mbox{and}\,\,\sigma_h(Fh^{-1}h')=Fh'.
    $$
   They imply that $\sigma_h\in \calN$ and then $h\in \calN_F$. This concludes the lemma.
  \end{proof}

\begin{prop}
\label{prop-pregaloisgroup}
 Let $\tilde{d}$ be as in Corollary \ref{cor-boundeddegree} of Appendix B
  and $\calH=\stab(Z_{F,\tilde{d}})$. Then $\calH$ is a proto-Galois group of (\ref{lodes}) and moreover, $Z_{F,\tilde{d}}=F\calH_F$.
\end{prop}
\begin{proof}
  We first show that $Z_{F,\tilde{d}}=F\calH_F$.
  Assume that $Z_{F,\tilde{d}}=FH$.
  If $H$ is a group, then one has that $\calH_F=H$ by Lemma \ref{lem-stabilizer}. That is to say, $Z_{F,\tilde{d}}=FH=F\calH_F$. Hence it suffices to show that $H$ is a group.
    As $F\in Z_{F,\tilde{d}}$, we have that $I_n\in H$. Suppose that $h_1,h_2\in H$. For any $P(X)\in I_{F,\tilde{d}}$, the equality $P(Fh_2)=0$ implies that $P(Xh_2)$ is an element of $I_{F,\tilde{d}}$. Hence $P(Fh_1h_2)=0$ and then $h_1h_2\in H$. It remains to prove that for any $h\in H$, $h^{-1}\in H$. As $H$ is closed under the multiplication, $Hh\subseteq H$. Multiplying both sides of $Hh\subseteq H$ by $h$ repeatedly, we have that
   $$
      \cdots \subseteq Hh^3\subseteq Hh^2\subseteq Hh \subseteq H.
   $$
   It is easy to verify that $H$ is a Zariski closed subset of $\GL_n(C)$ and so is $Hh^i$ for all positive integer $i$.
   The stability of the above sequence indicates that $Hh^{i_0+1}=Hh^{i_0}$ for some $i_0 \geq 0$ and therefore $Hh=H$. So $h^{-1}\in H$.

   Now we prove that $\calH$ is a proto-Galois group of (\ref{lodes}). First of all, as $Z_{F,\tilde{d}}$ is $k$-definable, $\calG\subseteq \stab(Z_{F,\tilde{d}})=\calH$. By Corollary \ref{cor-boundeddegree}, there is a proto-Galois group $\tilde{H}$ of (\ref{lodes}) bounded by $\tilde{d}$. Since $\calG_F\subseteq \tilde{H}$, $\calG\subseteq \stab(F\tilde{H})$ and then by Lemma~\ref{lem-definable}, $F\tilde{H}$ is a $k$-definable element of $N_{\tilde{d}}(V_{inv}^n)$. Then there are polynomials $Q_1(X), \cdots, Q_s(X)$ in
   $k[x_{1,1},x_{1,2},\cdots,x_{n,n}]_{\leq \tilde{d}}$ such that $F\tilde{H}=\Zero(Q_1(X),\cdots, Q_s(X))\cap V_{inv}^n$. It follows from $Q_i(F)=0$ that $Q_i\in I_{F,\tilde{d}}$ for all $i$ with $1\leq i \leq s$. Hence $F\calH_F=Z_{F,\tilde{d}}\subseteq F\tilde{H}$.
   This implies that $\calH_F\subseteq \tilde{H}$. Note that $(\calH_F^\circ)^t$ is generated by all unipotent elements of $\calH_F^\circ$ that are in $\tilde{H}^\circ$. Hence
   $$(\calH_F^\circ)^t\subseteq (\tilde{H}^\circ)^t\subseteq \calG_F^\circ\subseteq \calH_F^\circ.$$ Then the conclusion follows from the fact that $(\calH_F^\circ)^t$ is a normal subgroup of $\calH_F^\circ$.
\end{proof}

\section{Recovering Galois groups}
\label{sec-galoisgroup}
 Throughout this section, $I_{F,\tilde{d}}, Z_{F,\tilde{d}}$ and $\calH$ are as in Proposition \ref{prop-pregaloisgroup}. We will first compute a Zariski closed subset of $Z_{F,\tilde{d}}$ whose stabilizer is $\calG^\circ$. Then using the Galois group of finite extension, we construct $I_{\tilde{F}}$ and then the Galois group $\calG$, where $I_{\tilde{F}}$ is defined in (\ref{eqn-partialalgrelations}) with some fundamental matrix $\tilde{F}$ and $d=\infty$. Note that $\calG^\circ$ is defined as the pre-image of $\calG_F^\circ$ under the map $\phi_F$ in Section~\ref{sec-preliminary}. It is well-known that $\calG^\circ$ is equal to $\gal(\bar{k}K/\bar{k})$ where $\bar{k}K$ is the Picard-Vessiot extension field of $\bar{k}$ for (\ref{lodes}).

\subsection{Identity component $\calG^\circ$}
Decomposing $\calH_F$ into irreducible components, we obtain its identity component $\calH_F^\circ$. The defining equations of $\calH_F^\circ$ will lead to a Zariski closed subset $Z_{\calH^\circ}$ of $Z_{F,\tilde{d}}$ such that the stabilizer of $Z_{\calH^\circ}$ is $\calH^\circ$. Let $\chi_1,\cdots, \chi_l$ be generators of $X(\calH_F^\circ)$, where $X(\calH_F^\circ)$ is the group of characters of $\calH_F^\circ$. We will show that each character corresponds to a hyperexponential element over $\bar{k}$. Assuming we can find $\chi_1,\cdots,\chi_l$ (and we will show how this can be done), the results in \cite{compoint-singer} allow us to find algebraic relations among hyperexponential elements associated with $\chi_1,\cdots, \chi_l$. These relations together with $Z_{\calH^\circ}$ produce a Zariski closed subset $Z$ of $Z_{\calH^\circ}$ such that the identity component of $\stab(Z)$ is $\calG^\circ$. Using the similar argument as constructing $Z_{\calH^\circ}$, we are able to find a Zariski closed subset whose stabilizer is ${\calG^\circ}$.

Let $\bar{k}K$ be the Picard-Vessiot extension field of $\bar{k}$ for (\ref{lodes}) and $H$ a subgroup of $\GL_n(C)$. For brevity, we will use $H(\bar{k}K)$ to denote $\Zero(I(H))\cap \GL_n(\bar{k}K)$ where $I(H)$ is the vanishing ideal of $H$ in $C[x_{1,1},x_{1,2},\cdots, x_{n,n}]$. Let $\calN$ be a subgroup of $\GL_n(V)$. Suppose that $F\calN_F$ is $\bar{k}$-definable and $\Phi$ is the vanishing ideal of $F\calN_F$ in $\bar{k}[x_{1,1},x_{1,2},\cdots,x_{n,n}]$. Let $\gamma$ be an element in $\Zero(\Phi)\cap \GL_n(\bar{k})$. Then we have the following porposition.
\begin{prop}
  \label{prop-Hcirc}
  \begin{itemize}
   \item [$(a)$] For every $\bar{F}\in F\calN_F$, there is $g_{\bar{F}} \in\calN_F$ such that $$\left(\gamma g_{\bar{F}}\right)^{-1} \bar{F}\in \calN_F^\circ(\bar{k}K);$$
   \item [$(b)$] Let $g_{\bar{F}}$ be an element in $\calN_F$ such that $(a)$ holds for $\bar{F}$ and let $\alpha=\gamma g_{\bar{F}}$. Set
   $$
      Z_{\alpha}=\Zero\left(\left\{P(\alpha^{-1} X)\,\,|\,\, P(X)\in I(\calN_F^\circ)\right\} \right)\bigcap V_{inv}^n.
   $$
   Then $\stab(Z_{\alpha})=\calN^\circ$ and $Z_{\alpha}=\bar{F}\calN_F^\circ$, which is a Zariski closed subset of $F\calN_F$.
  \end{itemize}
\end{prop}
\begin{proof}
  $(a)$.  One can easily verify that $\Zero(\Phi)\cap\GL_n(\bar{k}K)=F\calN_F(\bar{k}K)$. Then we have that $\bar{F}\calN_F(\bar{k}K)=F\calN_F(\bar{k}K)$ because $\bar{F}\in F\calN_F$. Therefore $\gamma$ is an element of $ \bar{F}\calN_F(\bar{k}K)$. Equivalently, $\gamma^{-1}\bar{F}$ is in $\calN_F(\bar{k}K)$. Since $\calN_F$ is $C$-definable, there is $g_{\bar{F}}\in \calN_F$ such that
 $$
   g_{\bar{F}}^{-1}\gamma^{-1}\bar{F}=\left(\gamma g_{\bar{F}}\right)^{-1} \bar{F}\in \calN_F^\circ(\bar{k}K).
 $$

 $(b)$. Note that $V_{inv}^n=\{\bar{F}h \,\,|\,\, h\in \GL_n(C)\}$. From the definition of $Z_\alpha$, for any $\bar{F}h\in Z_{\alpha}$, we have that $\alpha^{-1} \bar{F} h$ belongs to $\calN_F^\circ(\bar{k}K)$. It means that $h\in \calN_F^\circ$, because $\alpha^{-1} \bar{F} \in \calN_F^\circ(\bar{k}K)$ and $h\in \GL_n(C)$. Therefore
 $Z_{\alpha}\subseteq \bar{F}\calN_F^\circ$. It is obvious that $\bar{F}\calN_F^\circ$ is a subset of $Z_{\alpha}$. Hence
 $$Z_{\alpha}=\bar{F}\calN_F^\circ\subseteq \bar{F}\calN_F=F\calN_F$$ that is a Zariski closed subset of $F\calN_F$.

 Finally, we show that $\stab(Z_{\alpha})=\calN^\circ$. Denote $\stab(Z_\alpha)$ by $\calW$. As $Z_\alpha=\bar{F}\calN_F^\circ$, Lemma \ref{lem-stabilizer} implies that $\calW_{\bar{F}}=\calN_F^\circ$. From the assumption, $\bar{F}=F\bar{h}$ for some $\bar{h}\in \calN_F$. Hence $\calW_F=\bar{h}\calW_{\bar{F}}\bar{h}^{-1}=\bar{h}\calN_F^\circ \bar{h}^{-1}$. Owing to the normality of $\calN_F^\circ$ in $\calN_F$, we have that $\calW_F=\calN_F^\circ$.  Therefore $\calW=\calN^\circ$.
\end{proof}

\begin{remark}
 Propositions \ref{prop-pregaloisgroup} and \ref{prop-Hcirc} enable us to compute a Zariski closed subset $Z_{\calH^\circ}$ of $Z_{F,\tilde{d}}$ such that $\stab(Z_{\calH^\circ})=\calH^\circ$. It suffices to compute an element $\alpha$ in $\Zero(I_{F,\tilde{d}})\cap \GL_n(\bar{k})$ and $\bar{F}\in F\calH_F$ satisfying that
     $\alpha^{-1}\bar{F}\in \calH_F^\circ(\bar{k}K)$.
\end{remark}
Let $\alpha$ be an element of $\Zero(I_{F,\tilde{d}})\cap \GL_n(\bar{k})$ and $\bar{F}$ in $F\calH_F$ satisfying that $\alpha^{-1}\bar{F}\in \calH_F^\circ(\bar{k}K)$. From the above proposition, we know that such $\alpha$ and $\bar{F}$ exist.
\begin{prop}
  \label{prop-toric}
  Let $\chi:\calH_F^\circ \rightarrow C^*$ be a character of $\calH_F^\circ$, which is represented by a polynomial in $C[x_{1,1},x_{1,2},\cdots, x_{n,n}, 1/\det(X)]$.
   Then $\chi(\alpha^{-1}\bar{F})$ is a hyperexponential element over $\bar{k}$. Moreover for any $h\in \calH_F^\circ$, $$\chi(\alpha^{-1}\bar{F}h)=\chi(\alpha^{-1}\bar{F})\chi(h).$$
  \end{prop}
  \begin{proof}
       Since $\chi$ is a character of $\calH_F^\circ$, for any $h_1,h_2\in \calH_F^\circ(\bar{k}K)$, $\chi(h_1h_2)=\chi(h_1)\chi(h_2)$. Note that $\alpha^{-1}\bar{F}\in \calH_F^\circ(\bar{k}K)$. For any $h\in \calH_F^\circ$,
     $$\chi(\alpha^{-1}\bar{F}h)=\chi(\alpha^{-1}\bar{F})\chi(h)=\chi(\alpha^{-1}\bar{F})\chi(h).$$
     Suppose that $\sigma\in\calG^\circ$. Then $\sigma(\alpha^{-1}\bar{F})=\alpha^{-1}\bar{F}[\sigma]$ for some $[\sigma]\in \GL_n(C)$. As $\alpha^{-1}\bar{F}$ belongs to $\calH_F^\circ(\bar{k}K)$ that is $C$-definable,
     we have that $\sigma(\alpha^{-1}\bar{F})\in \calH_F^\circ(\bar{k}K)$. It follows that $[\sigma]\in \calH_F^\circ$.
     Hence for any $\sigma\in \calG^\circ$,
     $$
        \sigma\left(\frac{\chi(\alpha^{-1}\bar{F})'}{\chi(\alpha^{-1}\bar{F})}\right)=
        \frac{\chi(\alpha^{-1}\bar{F}[\sigma])'}{\chi(\alpha^{-1}\bar{F}[\sigma])}=\frac{\chi(\alpha^{-1}\bar{F})'\chi([\sigma])}{\chi(\alpha^{-1}\bar{F})\chi([\sigma])}=\frac{\chi(\alpha^{-1}\bar{F})'}{\chi(\alpha^{-1}\bar{F})}.
     $$
     Thus $\frac{\chi(\alpha^{-1}\bar{F})'}{\chi(\alpha^{-1}\bar{F})} \in \bar{k}$.
  \end{proof}

  Suppose that $\chi_1,\cdots, \chi_l$ are the generators of $X(\calH_F^\circ)$, all of which are nontrivial and represented by polynomials in $C[x_{1,1},x_{1,2},\cdots, x_{n,n}]$.
  Then each character $\chi_i$ corresponds to a hyperexponential element $\chi_i(\alpha^{-1}\bar{F})$, denoted by $h_i$.
  Let $v_i=h_i'/h_i$ for all $i$ with $1\leq i \leq l$, and $E=\bar{k}(h_1,\cdots,h_l)$ that is the Picard-Vessiot extension of $\bar{k}$ for the equations
  $$\delta(Y)=\diag(v_1,\cdots, v_l)Y.$$
  Note that $E$ is a subfield of $\bar{k}K$.
   Let $\bfh=(h_1,\cdots,h_l)$. Then $\bfh$ is a fundamental matrix of the above equations and $\gal(E/\bar{k})$ can naturally be embedded into $(C^*)^l$. Denote the image of $\gal(E/\bar{k})$ by $T$ under this embedding. That is to say,
   $$
     T=\left\{\left.(c_1,\cdots, c_l)^T \in (C^*)^l\,\,\right | \,\,\exists \,\,\sigma\in \gal(E/\bar{k}) \,\,s.t.\,\, \sigma(h_i)=c_ih_i, i=1,\cdots, l\right\}.
   $$
  In \cite{compoint-singer}, the authors show that when $C$ is an algebraically closed computable field, given $v_1,\cdots, v_l$, one can compute
  a set of elements $S=\{h_{\eta_1},\cdots,h_{\eta_r}\}\subseteq \{h_1,\cdots, h_l\}$  such that
  \begin{itemize}
     \item [$(i)$]$h_{\eta_1},\cdots, h_{\eta_r}$ are algebraically independent over $C$;
     \item [$(ii)$]
       for each $j\in \{1,\cdots, l\}$, there are an element $f_j\in \bar{k}$ and integers $m_j, m_{i,j}, m_j\neq 0$ satisfying
      $$h_j^{m_j}=f_j\prod_{i=1}^r h_{\eta_i}^{m_{i,j}}$$
        if $S$ is nonempty, or $h_j^{m_j}=f_j$ if $S$ is empty.
  \end{itemize}
  The equalities in $(ii)$ generate almost all algebraic relations among $h_1,\cdots, h_l$.
  Due to the proof of Proposition 2.5 in \cite{compoint-singer}, the set $\{y_j^{m_j}-\prod_{i=1}^r y_{\eta_i}^{m_{i,j}}, j=1,2,\cdots, l\}$ defines an algebraic subgroup of $(C^*)^l$, whose identity component is equal to $T$.
   Let $\varphi=(\chi_1,\cdots, \chi_l)$. Then $\varphi$ is a surjective morphism from $\calH_F^\circ$ to $(C^*)^l$, for all $\chi_i$ are nontrivial. As $\bar{F}\in F\calH_F$, it follows from the normality of $\calH_F^\circ$ in $\calH_F$  that $\calH_{\bar{F}}^\circ=\calH_F^\circ$. Thus $\calG_{\bar{F}}^\circ\subseteq \calH_{\bar{F}}^\circ=\calH_F^\circ$. Moreover, we have
  \begin{lemma}
  \label{lem-toricpart}
    $ \varphi(\calG_{\bar{F}}^\circ)=T.$
  \end{lemma}
  \begin{proof}
      Note that $\calG^\circ=\gal(\bar{k}K/\bar{k})$ and $E\subseteq \bar{k}K$ that is a Picard-Vessiot extension field. By the Galois theory, the map $\psi:\calG^\circ \rightarrow \gal(E/\bar{k})$ given by
       $\psi(\sigma)=\sigma |_E$ for any $\sigma \in \calG^\circ$ is surjective.
       For any $[\sigma]\in \calG_{\bar{F}}^\circ$, there is $\sigma\in \calG^\circ$ such that $\sigma(\bar{F})=\bar{F}[\sigma]$ and then
       \begin{align*}
         \psi(\sigma)(\bfh) &=\sigma(\bfh)=(\sigma(h_1),\cdots,\sigma(h_l))=(\chi_1(\alpha^{-1}\bar{F}[\sigma]),\cdots,\chi_l(\alpha^{-1}\bar{F}[\sigma]))\\
         &=(\chi_1(\alpha^{-1}\bar{F})\chi_1([\sigma]),\cdots,\chi_l(\alpha^{-1}\bar{F})\chi_l([\sigma]))=(\chi_1([\sigma])h_1,\cdots, \chi_l([\sigma])h_l).
       \end{align*}
       By the definition of $T$, we have that $\varphi([\sigma])=(\chi_1([\sigma]),\cdots, \chi_l([\sigma]))\in T$. Now assume that $(c_1,\cdots, c_l)^T\in T$. Then there is $\sigma\in \gal(E/\bar{k})$ such that
        $\sigma(h_j)=c_jh_j$ for all $j$ with $1\leq j\leq l$. Due to the surjectivity of $\psi$, there is $\hat{\sigma}\in \calG^\circ$ such that $\psi(\hat{\sigma})=\sigma$. Assume that $\hat{\sigma}(\bar{F})=\bar{F}[\hat{\sigma}]$ for some $[\hat{\sigma}]\in \calG_{\bar{F}}^\circ$. It follows that
       $$
          c_j=\frac{\sigma(h_j)}{h_j}=\frac{\hat{\sigma}(h_j)}{h_j}=\frac{\chi_j(\alpha^{-1}\bar{F}[\hat{\sigma}])}{\chi_j(\alpha^{-1}\bar{F})}
          =\frac{\chi_j(\alpha^{-1}\bar{F})\chi_j([\hat{\sigma}])}{\chi_j(\alpha^{-1}\bar{F})}=\chi_j([\hat{\sigma}]), \,\,j=1,\cdots,l.
       $$
       In other words, $(c_1,\cdots, c_l)^T$ is the image of $[\hat{\sigma}]$ under the morphism $\varphi$. So $\varphi(\calG_{\bar{F}}^\circ)=T$.
  \end{proof}
  Set
   $$
     J=\{P(\alpha^{-1}X)\,\,|\,\,P(X)\in I(\calH_F^\circ)\}\bigcup\left\{\chi_j(\alpha^{-1}X)^{m_j}-f_j\prod_{i=1}^r \chi_{\eta_i}(\alpha^{-1}X)^{m_{i,j}}, j=1,\cdots,l\right\}.
   $$
   Let $Z_J=\zero(J)\cap V^n_{inv}$ and $\bar{\calH}=\stab(Z_J)$.
   \begin{prop}
   \label{prop-identitycomp}
      $\bar{\calH}^\circ=\calG^\circ$.
   \end{prop}
   \begin{proof}
   Assume that $Z_J=\bar{F}\bar{H}$ where $\bar{H}\subseteq \GL_n(C)$. Recall that $\alpha^{-1}\bar{F}\in \calH_F^\circ(\bar{k}K)$. It is easy to verify that
   \begin{equation*}
      \bar{H}=\calH_F^\circ \bigcap \zero\left(\left\{\chi_j^{m_j}-\prod_{i=1}^r \chi_{\eta_i}^{m_{i,j}}, \,\,j=1,2,\cdots, l\right\}\right).
   \end{equation*}
   Therefore $\bar{H}$ is a group and then by Lemma \ref{lem-stabilizer}, $\bar{\calH}_{\bar{F}}=\bar{H}$. To prove $\bar{\calH}^\circ=\calG^\circ$, it suffices to show that
   $\bar{H}^\circ=\calG_{\bar{F}}^\circ$.
   Since $\varphi=(\chi_1,\cdots, \chi_l)$ is surjective,
   $$\varphi(\bar{H})=(C^*)^l\bigcap\zero\left(\left\{y_j^{m_j}-\prod_{i=1}^r y_{\eta_i}^{m_{i,j}}, j=1,2,\cdots, l\right\}\right).$$
   The discussion before Lemma \ref{lem-toricpart} indicates that the identity component of $\varphi(\bar{H})$ is equal to $T$.
   Note that $\ker(\varphi)=(\calH^\circ_F)^t=(\calH_{\bar{F}}^\circ)^t$. Then we have $\ker(\varphi)\subseteq \calG_{\bar{F}}^\circ$, since $\calH$ is a proto-Galois group.
   As $Z_J$ is $\bar{k}$-definable, $\calG^\circ \subseteq \bar{\calH}$ and then $\calG_{\bar{F}}^\circ\subseteq \bar{\calH}_{\bar{F}}=\bar{H}$.
   Now we have
   $$[\bar{H}:\calG_{\bar{F}}^\circ]=\left[\bar{H}/\ker(\varphi):\calG_{\bar{F}}^\circ/\ker(\varphi)\right]=[\varphi(\bar{H}): \varphi(\calG_{\bar{F}}^\circ)].$$
   Owing to Lemma \ref{lem-toricpart}, $\varphi(\calG_{\bar{F}}^\circ)$ is equal to $T$ that is the identity component of $\varphi(\bar{H})$. Hence $[\bar{H}:\calG_{\bar{F}}^\circ]$ is finite. Thus  $\bar{H}^\circ=\calG_{\bar{F}}^\circ$, which completes the proof.
   \end{proof}
   Due to Proposition \ref{prop-Hcirc}, there are
   \begin{equation}
   \label{eqn-beta}
       \beta\in \Zero(J)\cap\GL_n(\bar{k})\,\,\mbox{and}\,\, \tilde{F}\in \bar{F}\bar{\calH}_{\bar{F}}
    \end{equation}
   such that $\beta^{-1}\tilde{F}\in \calG_{\bar{F}}^\circ(\bar{k}K)$.
   Let
   $$
      Z_\beta=\Zero\left(\left\{ P(\beta^{-1}X)\,\,|\,\,P(X)\in I(\calG_{\bar{F}}^\circ)\right\}\right)\cap V_{inv}^n.
   $$
   By Proposition \ref{prop-Hcirc} again, we have that $Z_\beta=\tilde{F}\calG_{\bar{F}}^\circ$ that is a Zariski closed subset of $\bar{F}\bar{\calH}_{\bar{F}}$ and $\stab(Z_\beta)=\calG^\circ$. Furthermore, by Lemma \ref{lem-stabilizer}, $\calG_{\tilde{F}}^\circ=\calG_{\bar{F}}^\circ$.

   \subsection{Galois group $\calG$}
   Let $\beta$ and $\tilde{F}$ be as in (\ref{eqn-beta}).
   Let $k_G$ be the Galois closure of $k(\beta^{-1})$, where $k(\beta^{-1})$ denotes the extension field of $k$ by joining the entries of $\beta^{-1}$.
   For any $\tau\in \gal(k_G/k)$, set
   $$
       J_{\tau(\beta)}=\left\langle\left\{P(\tau(\beta)^{-1}X)\,\,|\,\, P(X)\in I(\calG_{\tilde{F}}^\circ)\right\}\right\rangle
   $$
   where $\langle * \rangle$ denotes the ideal in $k_G[x_{1,1},x_{1,2},\cdots, x_{n,n}]$ generated by $*$.
   \begin{prop}
    \label{prop-algebraicrelations}
    Let $I_{\tilde{F}}$ be defined in (\ref{eqn-partialalgrelations}) with $F=\tilde{F}$ and $d=\infty$ . Then
    $$
        I_{\tilde{F}}=\left(\bigcap _{\tau\in \gal(k_G/k)} J_{\tau(\beta)}\right)\bigcap k[x_{1,1},x_{1,2},\cdots, x_{n,n}].
    $$
   \end{prop}
   \begin{proof}
     Denote the ideal in the righthand side of the above equality by $\Phi$. Suppose that $P(X)\in \Phi$. Then $P(X)\in J_\beta$. From the previous subsection, $\tilde{F}\in Z_\beta$. So $P(\tilde{F})=0$. That is to say, $P(X)\in I_{\tilde{F}}$. Thus $\Phi\subseteq I_{\tilde{F}}$. Conversely,
     suppose that $P(X)\in I_{\tilde{F}}$. Then $P(\tilde{F})=0$. Applying $\calG^\circ$ to it, we obtain that $P(\tilde{F}g)=0$ for every $g\in \calG_{\tilde{F}}^\circ$ and furthermore the polynomial $P(\tilde{F}X)$ vanishes on $\calG_{\tilde{F}}^\circ(\bar{k}K)$. Since $\beta^{-1}\tilde{F}\in \calG_{\tilde{F}}^\circ(\bar{k}K)$, $\tilde{F}\calG_{\tilde{F}}^\circ(\bar{k}K)=\beta\calG_{\tilde{F}}^\circ(\bar{k}K)$. It implies that $P(\beta X)$ vanishes on $\calG_{\tilde{F}}^\circ(\bar{k}K)$. Consequently, $P(\beta X)$ belongs to $\langle I(\calG_{\tilde{F}}^\circ)\rangle$ and then $P(X)\in J_\beta$. Because all coefficients of $P(X)$ are in $k$, $P(X)\in J_{\tau(\beta)}$ for all $\tau\in \gal(k_G/k)$. Hence $P(X)\in \Phi$.
   \end{proof}
   If $I_{\tilde{F}}$ is computed, then it is easy to verify that
   $$
      \calG_{\tilde{F}}=\{g\in \GL_n(C)\,\,|\,\, \forall \,\,P(X)\in I_{\tilde{F}}, \,\,P(Xg)\in I_{\tilde{F}}\}.
   $$
   In the following, we present another method to compute $\calG_{\tilde{F}}$ that avoids computing $I_{\tilde{F}}$.
   From Proposition 3.20 and Theorem 3.11 of \cite{magid}, there is a Picard-Vessiot extension field, say $\tilde{K}$, of $k$ that contains $k_G$ and $K$ as subfields. Then Galois theory implies that if $\tau$ is an element of $\gal(k_G/k)$ (or $\calG$), then there is $\rho\in \gal(\tilde{K}/k)$ such that the restriction of $\rho$ on $k_G$ (or $K$) is equal to $\tau$. Let $\{Q_1(X),\cdots, Q_\nu(X)\}$ be a set of generators of $I(\calG_{\bar{F}}^\circ)$. Set
    $$
       \tilde{G}=\bigcup_{\tau\in \gal(k_G/k)}\left\{\,\, g\in \GL_n(C) \left| \,\, Q_i(\tau(\beta)^{-1}\tilde{F}g)=0, \,\,\mbox{for all $i=1,\cdots, \nu$}\right.\right\}.
    $$
    Then we have
    \begin{prop}
    \label{prop-finitepart}
      $\calG_{\tilde{F}}=\tilde{G}$.
    \end{prop}
    \begin{proof}
       We first prove that $\calG_{\tilde{F}}\subseteq \tilde{G}$. Assume that $[\sigma]\in \calG_{\tilde{F}}$. Then there is $\sigma\in \calG$ such that $\sigma({\tilde{F}})={\tilde{F}}[\sigma]$ and furthermore there is $\rho\in \gal(\tilde{K}/k)$ such that $\rho|_K=\sigma$. Let $\tau=\rho|_{k_G}$. Note that $\tau\in \gal(k_G/k)$. Then for all $i$ with $1\leq i\leq \nu$,
       $$
          \rho(Q_i(\beta^{-1}{\tilde{F}}))=Q_i(\tau(\beta)^{-1}\sigma({\tilde{F}}))=Q_i({\tau(\beta)^{-1}\tilde{F}}[\sigma])=0.
       $$
        This implies that $[\sigma]\in \tilde{G}$. Conversely, assume that $g\in \tilde{G}$. Then there is $\tau\in \gal(k_G/k)$ such that $Q_i(\tau(\beta)^{-1}{\tilde{F}}g)=0$ for all $i$ with $1 \leq i \leq \nu$. There is $\rho\in \gal(\tilde{K}/k)$ such that $\rho |_{k_G}=\tau^{-1}$. Since $\rho|_K\in \calG$, there is $h\in \calG_{\tilde{F}}$ such that $\rho |_K({\tilde{F}})={\tilde{F}}h$. Now we have that
        $$\forall\,\,i=1,\cdots, \nu,\,\,\rho(Q_i(\tau(\beta)^{-1}{\tilde{F}}g))=Q_i(\beta^{-1}{\tilde{F}}hg)=0. $$
         Hence $\beta^{-1}{\tilde{F}}hg\in \calG_{\tilde{F}}^\circ(\tilde{K})$, which implies that $hg\in \calG_{\tilde{F}}^\circ$ and thus $g\in \calG_{\tilde{F}}$.
    \end{proof}

  \section{Simplified Hrushovski's Algorithm}
   Now we are ready to present the algorithm. Instead of $\calG$, we shall compute $\calG_{\tilde{F}}$ for some fundamental matrix $\tilde{F}$. Throughout this section, $C$ will denote the algebraic closure of a computable field of characteristic zero.
   \begin{algorithm}
   Input: A linear differential equation (\ref{lodes}) with coefficients in $C(t)$. \\
   Output: the Galois group $\calG_{\tilde{F}}$.
   \begin{itemize}
      \item [$(a)$] By Corollary \ref{cor-boundeddegree}, determine an integer $\tilde{d}$ such that there is a proto-Galois group of (\ref{lodes}) bounded by $\tilde{d}$.
      \item [$(b)$] Compute a fundamental matrix $F$ (the first finitely many terms of its formal power series expansion at some point). Compute $I_{F, \tilde{d}}$ and then $\calH_F$, where $I_{F,\tilde{d}}$ is defined in (\ref{eqn-partialalgrelations}) and $\calH$ is the stabilizer of $\Zero(I_{F,\tilde{d}})\cap V_{inv}^n$.
      \item [$(c)$] Compute $\calH_F^\circ$ and find a zero $\alpha$ of $I_{F,\tilde{d}}$ in $\GL_n(\overline{C(t)})$ and $\bar{F}\in F\calH_F$ such that $\alpha^{-1}\bar{F}$ is an element of
      $\calH_F^\circ(\bar{k}K)$.
     \item [$(d)$]
         Compute generators of $X(\calH_F^\circ)$, say $\chi_1,\cdots, \chi_l$, which are represented by polynomials in $C[x_{1,1},x_{1,2},\cdots, x_{n,n}]$. Denote $\chi_i(\alpha^{-1}\bar{F})$ by $h_i$
         for $i=1,2,\cdots, l$. By Proposition \ref{prop-Hcirc},
         $h_i$ is hyperexponential over $\overline{C(t)}$. Compute these $h_i$.
       \item [$(e)$] Using the method in \cite{compoint-singer}, compute algebraic relations for $h_1,\cdots,h_l$, say
         $$
            \left\{h_j^{m_j}-f_j\prod_{i=1}^r h_{\eta_i}^{m_{i,j}}, j=1,2,\cdots, l\right\}
         $$
         where $f_j\in \overline{C(t)}$ and $h_{\eta_1},\cdots,h_{\eta_r}$ are algebraically independent over $C(t)$.
     \item [$(f)$]
        Denote a set of generators of $I(\calH_F^\circ)$ by $\{P_1(X),\cdots, P_\mu(X)\}$ and set
        $$
           J=\{P_i(\alpha^{-1}X),i=1,\cdots, \mu\}\bigcup \left\{ \chi_j^{m_j}(\alpha^{-1}X)-f_j\prod_{i=1}^r \chi_{\eta_i}^{m_{i,j}}(\alpha^{-1}X), j=1,\cdots, l \right\}.
        $$
        Compute $\bar{\calH}=\stab(\Zero(J)\cap V_{inv}^n)$ and then $\bar{\calH}^\circ$ that is equal to $\calG^\circ$.
     \item [$(g)$]As in the step $(c)$, compute $\calG_{\bar{F}}^\circ$ and find
        $\beta$ in $\Zero(J)\cap\GL_n(\overline{C(t)})$ and $\tilde{F}\in \bar{F}\bar{\calH}_{\bar{F}}$ such that $\beta^{-1}\tilde{F}\in\calG_{\bar{F}}^\circ(\bar{k}K)$.
     \item [$(h)$] Denote a set of generators of $I(\calG_{\bar{F}}^\circ)$ by $\{Q_1(X),\cdots, Q_\nu(X)\}$.
          Let $k_G$ be the Galois closure of $C(t)(\beta^{-1})$, where $C(t)(\beta^{-1})$ is the extension field of $C(t)$ by joining the entries of $\beta^{-1}$. Compute $\gal(k_G/C(t))$ and
         $$
            \calG_{\tilde{F}}=\bigcup_{\tau\in \gal(k_G/C(t))}\left\{g\in \GL_n(C)\,\,|\,\,\forall \,\,i=1,\cdots, \nu,\,\, Q_i(\tau(\beta)^{-1}\tilde{F}g)=0 \right\}.
         $$
     \item [$(i)$] Return $\calG_{\tilde{F}}$.
   \end{itemize}
   \end{algorithm}
   The correctness of the algorithm follows from the results in Sections 2 and 3. In the following, we will present several computation details omitted in the previous sections.
   Generally, it is difficult to find a fundamental matrix of (\ref{lodes}). What we can compute is the first finitely many terms of formal power series solutions of (\ref{lodes}) at some point of $C$. Let $z$ be a generic point of $C$. Expanding $A$ at $t=z$, we have
   \begin{equation}
   \label{eqn-expansionA}
     A=A_0+A_1(t-z)+A_2(t-z)^2+\cdots, \,\,A_i\in \Mat_n\left(C\left[z,\frac{1}{q(z)}\right]\right)
   \end{equation}
   where $q(z)$ is a polynomial in $C[z]$ such that $q(t)$ is the least common multiple of the denominators of the entries of $A$. Using the above expansion, we can compute
    a formal power series solutions of (\ref{lodes}) that has the following form
   $$
       \Gamma_z=I_n+D_1(t-z)+D_2(t-z)^2+\cdots,  \,\,D_i\in \Mat_n\left(C\left[z,\frac{1}{q(z)}\right]\right).
   $$
     If $a$ is an element of $C$ that does not vanish $q(z)$, then $\Gamma_z$ can be specialized to $\Gamma_a$ that is the formal power series expansion of some fundamental matrix of (\ref{lodes}) at $t=a$. In this section, we will use $F_a$ to denote the fundamental matrix of (\ref{lodes}) in $\GL_n(K)$ whose formal power series expansion at $t=a$ is of the form $\Gamma_a$.  We will begin with a fundamental matrix $F_a$ and then may replace it by other fundamental matrix if necessary during the computation process.


\subsection{Computing $I_{{F_a},\tilde{d}}$ and $\calH_{F_a}$}
\label{subsec-computpregaloisgroup}
 Corollary \ref{cor-bound-k-definable} says that the coefficients of defining polynomials in
$C(t)[x_{1,1},\cdots, x_{n,n}]_{\leq \tilde{d}}$ of $Z_{F,\tilde{d}}$ can be chosen to be rational functions bounded by an integer $\ell$. Without loss of generality, we may assume that all these coefficients are polynomial in $t$ which are bounded by $2\ell$. Let
$$
    P_{\bfc}(X)=\sum_{|\vec{m}|\leq \tilde{d}}\left(\sum_{0\leq i\leq 2\ell} c_{i,\vec{m}}(t-a)^i\right)X^{\vec{m}}, \,\,\bfc=(\cdots,c_{i,\vec{m}},\cdots).
$$
where the $c_{i,\vec{m}}$ are indeterminates.
A small modification of Theorem 1 in \cite{bertrand-beukers} yields the following theorem that bounds the order of $P_{\bfc}(F_a)$.
\begin{theorem}
\label{thm-bertrand}
One can compute an integer $N$ depending on $A,n,\ell$ such that
$$P_\bfc(F_a)=0\,\,\mbox{or}\,\,ord_{t=a}(P_\bfc(F_a))\leq N.$$
\end{theorem}
\begin{proof}
 Consider $F_a$ as a vector with $n^2$ entries. Then $F_a$ is a solution of the system (\ref{n-copy-lode}). The system (\ref{n-copy-lode}) is defined in Appendix A. Then the theorem follows from Theorem 1 in \cite{bertrand-beukers}.
\end{proof}
The above theorem can be generalized into the case that the coefficients of $P_{\bfc}(X)$ involve algebraic functions. More precisely, let $\gamma$ be an algebraic function with minimal polynomial $Q(x)$ and let $l=\deg(Q(x))$. Assume that $\bfw=(1,\gamma,\gamma^2, \cdots, \gamma^{l-1})^T$. Then it is easy to see that $\bfw$ is a solution of linear differential equations $\delta(Y)=BY$ where $B\in \Mat_l(C(t))$ can be constructed from $Q(x)$.
Let
$$
   \tilde{P}_{\bfc}(X)=\sum_{|\vec{m}|\leq \tilde{d}}\left(\sum_{j=1}^l\left(\sum_{0\leq i\leq 2\ell} c_{i,j,\vec{m}}(t-a)^i\right)\gamma^j\right)X^{\vec{m}}, \,\,\bfc=(\cdots,c_{i,j,\vec{m}},\cdots).
$$
Assume that $t=a$ is a regular point of $\gamma$. Then
\begin{cor}
\label{cor-bound}
There is an integer $\tilde{N}$ depending on $A, Q(x), n, \ell$ such that
$$
   \tilde{P}_\bfc(F_a)=0\,\,\mbox{or}\,\,\ord_{t=a}\tilde{P}_{\bfc}(F_a)\leq \tilde{N}.
$$
\end{cor}
\begin{proof}
  We only need to consider the system
  $$
     \delta(Y)=\diag(\underbrace{A,\cdots, A}_{n}, B)Y.
  $$
  Then the corollary follows from the above theorem.
\end{proof}
 Let $\calS$ be the set of the coefficients of the first $N+2$ terms of $P_{\bfc}(F_a)$. $\calS$ is a linear system in $\bfc$ and
$$
    P_{\bar{\bfc}}(F_a)=0 \Leftrightarrow \,\,\mbox{$\bar{\bfc}$ is a solution of $\calS$}.
$$
So computing $I_{{F_a},\tilde{d}}$ is reduced into solving the linear system $\calS$.

Assume that $I_{{F_a},\tilde{d}}$ is computed. Then we can compute $\calH_{F_a}$ as follows. For any $h\in \calH_{F_a}$, ${F_a}h\in Z_{{F_a},\tilde{d}}$. It implies that for any $P(X)\in I_{{F_a},\tilde{d}}$, $P(F_a h)=0$ and so that $P(Xh)\in I_{F_a,\tilde{d}}$.
This induces the defining equations for $\calH_{F_a}$. More precisely, let $P_1(X),\cdots, P_\mu(X)$ be a $C(t)$-basis of $I_{F_a,\tilde{d}}$. Let $\bfx=(\cdots, X^{\vec{m}},\cdots)$ be a vector consisting of all monomials in $x_{1,1},\cdots, x_{n,n}$ with degree not greater than $\tilde{d}$,
where $X^{\vec{m}}=x_{1,1}^{m_{1,1}}x_{1,2}^{m_{1,2}}\cdots x_{n,n}^{m_{n,n}}$. For any $h\in\GL_n(C)$, there is $[h]\in \GL_{n^2+\tilde{d} \choose \tilde{d}}(C)$ such that
$$h\cdot\bfx=(\cdots, (Xh)^{\vec{m}},\cdots)=\bfx[h].$$
For any $i=1,\cdots,\mu$, there is $\bfp_i \in C(t)^{n^2+\tilde{d} \choose \tilde{d}}$ such that $P_i(X)=\bfx\bfp_i$. Then we have
$$
  h\in \calH_{F_a}\Leftrightarrow \forall i, P_i(Xh)=\bfx[h]\bfp_i\in I_{F_a,\tilde{d}}\Leftrightarrow \forall i, [h]\bfp_i \wedge \bfp_1 \wedge \bfp_2 \wedge \cdots \wedge \bfp_\mu=0.
$$
This induces the defining equations for $\calH_{F_a}$.

\subsection{Computing $\alpha$ and $\bar{F}$}
\label{subsec-alpha}
Assume that we have calculated $I_{F_a,\tilde{d}}$ and $\calH_{F_a}$. Decomposing $\calH_{F_a}$ into irreducible components, we obtain $\calH_{F_a}^\circ$.
 Compute a zero $\bar{\alpha}$ of $I_{F_a, \tilde{d}}$ in $\GL_n(\overline{C(t)})$. Assume that $b\in C$ is a regular point of $\bar{\alpha}^{-1}$ and $q(b)\neq 0$, where $q(z)$ is as in (\ref{eqn-expansionA}). It is well-known that there is $g\in \GL_n(C)$ such that $F_a=F_b g$. In general, it is difficult to find such $g$, because we can only compute the first finitely many terms of $\Gamma_a$ and $\Gamma_b$. Note that $\Gamma_a$ (resp. $\Gamma_b$) is the formal power sereis expansion of $F_a$ (resp. $F_b$) at $t=a$ (resp. $t=b$). Fortunately, it is easy to find $h\in \GL_n(C)$ such that
 $F_b h$ belongs to $F_a\calH_{F_a}$. Such $h$ can be calculated as below. $F_bh\in F_a\calH_{F_a}$ if and only if
 for any $P(X)\in I_{F_a, \tilde{d}}$, $P(F_bh)=0$. Using the bound given in Theorem \ref{thm-bertrand}, the latter conditions derive the defining equations for $h$. Assume that we have found such $h$. Let $\bar{F}=F_bh$. Proposition \ref{prop-Hcirc} tells us that there is $\bar{g}\in \calH_{F_a}$ such that
  $(\bar{\alpha}\bar{g})^{-1}\bar{F}\in \calH_{F_a}^\circ(\bar{k}K)$. Since both $\bar{\alpha}^{-1}$ and $\bar{F}$ can be expanded as formal power series at the point $t=b$, the bound given in Theorem \ref{thm-bertrand} allows us to derive the defining equations for $\bar{g}$ from the condition $(\bar{\alpha}\bar{g})^{-1}\bar{F}\in \calH_{F_a}^\circ(\bar{k}K)$. Compute such $\bar{g}$ and let $\alpha=\bar{\alpha}\bar{g}$.

 \subsection{Computing $h_i$}
 By Remark \ref{rem-characters} in Appendix B, we can find a set of generators of $X(\calH_{F_a}^\circ)$, say $\chi_1,\cdots, \chi_l$.
  Assume that $\chi_1,\cdots, \chi_l$ are defined by polynomials of degree not greater than an integer $\kappa$. Denote the monomials in entries of $\alpha^{-1}\bar{F}$ with degree not greater than $\kappa$ by $\bfm_1,\cdots, \bfm_m$. Then
  each $\bfm_i$ satisfies a linear differential operator $L_i$ with coefficients in $C(t)(\alpha^{-1})$. Let $L=LCLM(L_1,\cdots,L_m)$. For each $i=1,\cdots,l$, $h_i=\chi_i(\alpha^{-1}\bar{F})$ is a hyperexponential solution of $L$. To compute $h_i$, it means to calculate $v_i=h_i'/h_i$ where $i=1,\cdots,l$. From \cite{davenport-singer}, one can compute all hyperexponential solutions of $L$ and then the bounds for minimal polynomials of $\bar{h}^{\prime}/\bar{h}$ where $\bar{h}$ is any hyperexponential solution. Using these bounds and Hermite-Pad\'{e} approximation, one can recover $v_i$ from the series expansion of $\chi_i(\alpha^{-1}\bar{F})'/\chi_i(\alpha^{-1}\bar{F})$.

\subsection{Computing $\calG_{\tilde{F}}$}
 The method described in Section \ref{subsec-alpha} can be adapted to find $\beta$ and $\tilde{F}$ in Step (g). So far, in Step (g), we obtain $\beta$ and $\tilde{F}$ satisfying that $\beta^{-1}\tilde{F}\in \calG_{\bar{F}}^\circ(\bar{k}K)$,  and a set of generators of $I(\calG_{\bar{F}}^\circ)$, denoted by $\{Q_1(X),\cdots, Q_\nu(X)\}$. Assume that $t=c$ is the point we pick to find $\beta$ and $\tilde{F}$ in Step (g). That is to say, both $\beta^{-1}$ and $\tilde{F}$ can be expanded as formal power series at the point $t=c$.
By Corollary \ref{cor-bound}, there is an integer $\tilde{N}$ such that for all $i$ with $1\leq i\leq \nu$ and $\tau\in \gal(k_G/C(t))$,
$$ Q_i(\tau(\beta)^{-1}\tilde{F} \bfc)=0\quad\mbox{or}\quad\ord_{t=c}(Q_i(\tau(\beta)^{-1}\tilde{F} \bfc))\leq \tilde{N}.$$
For each $\tau\in \gal(k_G/C(t))$, let $\calS_\tau$ be the set of coefficients of the first $\tilde{N}+2$ terms of formal power series $Q_1(\tau(\beta)^{-1}\tilde{F}\bfc), \cdots, Q_\nu(\tau(\beta)^{-1}\tilde{F}\bfc)$.
Then for each $\tau\in\gal(k_G/k)$, $\calS_\tau$ is the set of polynomials in $\bfc$ and for any $g\in \GL_n(C)$,
$Q_i(\tau(\beta)^{-1}\tilde{F}g)=0$ for all $i$ with $1\leq i\leq \nu$ if and only if $g\in \Zero(\calS_\tau).$
Let
$$
  \calG_{\tilde{F}}=\left(\bigcup_{\tau\in \gal(k_G/k)} \Zero(\calS_\tau)\right)\bigcap \GL_n(C).
$$
We then obtain the desired Galois group.
\begin{appendix}
\section{Bounds for $k$-definable elements of $N_d(V^n_{inv})$}
In this appendix, the symbols in the previous sections are used. Let $W$ be a $C$-vector subspace of $V$ of dimension not greater than $n$. Assume that $\bfw_1,\cdots, \bfw_m$ is a basis of $W$.
\begin{define}
    $W$ is said to be $k$-definable if there is $M\in GL_n(k)$ such that
   \begin{equation}
   \label{k-definable}
      M (\bfw_1,\bfw_2,\cdots,\bfw_m)=\begin{pmatrix} \tilde{W} \\ 0 \end{pmatrix}
   \end{equation}
   where $\tilde{W}$ is an invertible $m \times m$ matrix. We call $M$ the defining matrix of $W$.
\end{define}
\begin{remark}
\label{rem-stabilizer}
  There is some $B\in \mat_m(k)$ such that $\tilde{W}$ is a fundamental matrix of linear differential equations $\delta(Y)=BY$.
\end{remark}
 Clearly, $V$ is a $\calG$-module. Moreover, we have
\begin{prop}
\label{prop-definable}
   Let $W$ be a $C$-vector subspace of $V$. Then the following are equivalent:
   \begin{enumerate}
      \item [$(a)$]
        $W$ is $k$-definable;
   \item [$(b)$]
     $W$ is a $\calG$-submodule of $V$.
   \end{enumerate}
\end{prop}
\begin{proof}

   $(a)\Rightarrow (b)$. Assume that $W$ is $k$-definable and $\bfw_1,\cdots, \bfw_m$ is a $C$-basis of $W$. Then there is an $M \in \GL_n(k)$ such that
   (\ref{k-definable}) holds. For any $\sigma\in \calG$, it follows from Remark \ref{rem-stabilizer} that there is $[\sigma]\in \GL_m(C)$ such that $\sigma(\tilde{W})=\tilde{W} [\sigma]$. Applying $\sigma$ to (\ref{k-definable}), we have
   $$
      M (\sigma(\bfw_1),\sigma(\bfw_2),\cdots, \sigma(\bfw_m))=\begin{pmatrix} \tilde{W} \\ 0 \end{pmatrix}[\sigma].
   $$
   Then the above equality and (\ref{k-definable}) deduce that
   $$
      (\sigma(\bfw_1),\sigma(\bfw_2),\cdots, \sigma(\bfw_m))=(\bfw_1,\bfw_2,\cdots,\bfw_m)[\sigma].
   $$
   Hence $W$ is a $\calG$-submodule.

   $(b)\Rightarrow (a)$. Assume that $W$ is a $\calG$-submodule and $\bfw_1,\cdots, \bfw_m$ is a $C$-basis of $W$. Then there is an $m\times m$ submatrix of $(\bfw_1,\cdots, \bfw_m)$ which is invertible. Without loss of generality, assume that the submatrix consisting of the first $m$ rows of the matrix $(\bfw_1,\cdots, \bfw_m)$ is invertible. Denote this submatrix by $\tilde{W}$ and the left $(n-m)\times m$ submatrix by $\bar{W}$. Let $\bar{M}=\bar{W}\tilde{W}^{-1}$. Assume that $\sigma \in \calG$. We then have that $\sigma((\bfw_1,\cdots, \bfw_m))=(\bfw_1,\cdots, \bfw_m) [\sigma]$ for some $[\sigma] \in \GL_m(C)$. Applying $\sigma$ to $\bar{M}$ yields that
   $$\sigma(\bar{M})=\sigma(\bar{W}\tilde{W}^{-1})=\sigma(\bar{W})\sigma(\tilde{W}^{-1})=\bar{W}[\sigma](\tilde{W}[\sigma])^{-1}=\bar{W} \tilde{W}^{-1}=\bar{M},$$
   which implies that $\sigma(\bar{M})=\bar{M}$ for all $\sigma\in \calG$. Hence all entries of $\bar{M}$ belong to $k$. Let
   \begin{equation}
\label{eqn-form}
      M=\begin{pmatrix} I_m & 0 \\ -\bar{M} & I_{n-m}    \end{pmatrix},
\end{equation}
   Then $M$ is a matrix such that (\ref{k-definable}) holds.
\end{proof}
\begin{define}
 Let $W$ be a $k$-definable subspace of $V$ and $m$ an integer.  We call $W$ bounded by $m$ if there is $M$ with $\deg(M)\leq m$ such that (\ref{k-definable}) holds, where $\deg(M)$ is defined as the maximum of the degrees of entries of $M$.
\end{define}

\begin{prop}
\label{prop-bounded-definable-subspace}
There is an integer $\ell$ such that all $k$-definable subspaces of $V$ are bounded by $\ell$.
\end{prop}
 \begin{proof}
Evidently, it is enough to prove Proposition \ref{prop-bounded-definable-subspace} for the fixed dimension $k$-definable subspaces.
 Let $W$ be a $k$-definable subspace of $V$ with a basis $\bfw_1,\cdots,\bfw_m$.
From the proof of Proposition~\ref{prop-definable}, $M$ can be chosen to satisfy that there is a permutation matrix $P$ such that $MP$ is of the form (\ref{eqn-form}). In other words,
 $$
    (\bfw_1,\bfw_2,\cdots,\bfw_m)=P\begin{pmatrix} I_m & 0 \\ \bar{M} & I_{n-m}    \end{pmatrix}\begin{pmatrix} \tilde{W} \\ 0 \end{pmatrix},
 $$
 where $\bar{M}$ is an $(n-m)\times m$ matrix with entries in $k$. If $\deg(\bar{M})$ is not greater than $\ell$, neither is $\deg(M)$.
 One can easily see that the coordinates of $P\bfw_1\wedge P\bfw_2 \wedge \cdots \wedge P\bfw_m$ are the permutation of those of $\bfw_1\wedge \bfw_2 \wedge \cdots \wedge \bfw_m$.
 This implies that we only need to consider the case $P=I_n$. An easy calculation yields that
 \begin{align*}
    \bfw_1\wedge \bfw_2 \wedge \cdots \wedge \bfw_m &=\sum_{1\leq i_1<i_2<\cdots <i_m \leq n}  c_{i_1,i_2,\cdots, i_m} \bfe_{i_1}\wedge \bfe_{i_2}\wedge \cdots \wedge \bfe_{i_m}\\
    &=\sum_{1\leq i_1<i_2<\cdots <i_m \leq n}  a_{i_1,i_2,\cdots, i_m}\det(\tilde{W}) \bfe_{i_1}\wedge \bfe_{i_2}\wedge \cdots \wedge \bfe_{i_m},
 \end{align*}
 where $ c_{i_1,i_2,\cdots, i_m}\in K, a_{i_1,i_2,\cdots, i_m}\in k$. Note that $\delta(\tilde{W})=B\tilde{W}$ for some $B\in \Mat_m(k)$. Therefore $\delta(\det(\tilde{W}))=\Tr(B)\det(\tilde{W})$. It implies that the vector $(\cdots, c_{i_1,i_2,\cdots, i_m},\cdots)$ is a hyperexponential solution over $k$ of the exterior system $\delta(Y)=\left(\wedge^m A\right)Y$. The matrix $\wedge^m A$ can be constructed from $A$. Assume that $\bar{M}=(\lambda_{i,j}), 1\leq i\leq n-m, 1\leq j \leq m$. We further have that
 $$
    a_{12\cdots m}=1,\,\, a_{12\cdots \hat{j} \cdots m i}=\lambda_{i-m,j}, \,\, m+1\leq i \leq n, 1\leq j \leq m.
 $$

 Assume that $\bfv h$ is a hyperexponential solution of $\delta(Y)=\left(\wedge^m A\right)Y$ where $h$ is a hyperexponential element over $k$ and $\bfv\in k^{n \choose m}$. Note that if we fix hyperexponential elements $h$, then from \cite{hoeij-weil,put-singer}, there is an integer $\ell$ such that $\deg(\bfv)\leq\ell/2$. To obtain the entries of $\bar{M}$ from those of $\bfv$, we need to normalize $\bfv$ at some coordinate, for instance when $P=I_n$, $a_{12\cdots m}$ is normalized to 1. Hence $\deg(\bar{M})\leq \ell$.
 \end{proof}

 Since we restrict ourselves to Zariski closed subsets of $V^n_{inv}$, we need to consider the following two new systems.
One is the $n$-direct sum of (\ref{lodes}):
\begin{equation}
\label{n-copy-lode}
    \delta(Y)=\diag(\underbrace{A,A,\cdots,A}_{n})Y=A^{\oplus n}Y .\tag{$\dsA$}
\end{equation}
Then $V^n$ is the solution space of (\ref{n-copy-lode}) and
$\diag(\underbrace{F,F\cdots,F}_{n})$, denoted by $F^{\oplus n}$, is a fundamental matrix.
The other is the symmetric power system of (\ref{n-copy-lode}). Define a map
\begin{align*}
S^{\leq d}: V^n &\longrightarrow K^{n^2+d \choose d} \\
\bfv=(v_{1,1},v_{2,1},\cdots, v_{n,n})&\longrightarrow (\cdots, v_{1,1}^{\mu_{1,1}}v_{1,2}^{\mu_{1,2}}\cdots v_{n,n}^{\mu_{n,n}},\cdots,), \,\,\sum_{1\leq i,j \leq n}\mu_{i,j}\leq d.
\end{align*}
Then $S^{\leq d}(\bfv)$ is a solution of the symmetric power system of (\ref{lodes}), denoted by
\begin{equation}
\label{n-symmetric-system}
   \delta(Y)=\left(\ssA\right)Y. \tag{$\ssA$}
\end{equation}
The matrix $\ssA$ can be constructed from $A$ and its entries are in $k$.
Denote the solution space of (\ref{n-symmetric-system}) by $\ssV$. More details about the symmetric power system could be found in (p. 39, \cite{put-singer}) and \cite{hoeij-weil}.

Let $Z$ be an element of $N_d(V^n_{inv})$. Set
$$
  W_Z=\sp_C\{S^{\leq d}(\bfz), \bfz\in Z\}.
$$
Then $W_Z$ is a subspace of $\ssV$.
\begin{prop}
\label{prop-definable-corresponding}
Given an element $Z$ of $N_d(V^n_{inv})$, $Z$ is $k$-definable if and only if $W_Z$ is $k$-definable.
\end{prop}
\begin{proof}
Suppose that $S^{\leq d}(\bfv_1), \cdots, S^{\leq d}(\bfv_m)$ is a basis of $W_Z$.

$(\Rightarrow)$ Assume that $\sigma \in \calG$. Since $Z$ is $k$-definable, $\sigma(\bfv)\in Z$ for all $\bfv\in Z$. It implies that for each $i=1,\cdots,m$, $\sigma(S^{\leq d}(\bfv_i))=S^{\leq d}(\sigma(\bfv_i))\in W_Z$.
Consequently, $W_Z$ is a $\calG$-module. By Proposition \ref{prop-definable}, $W_Z$ is $k$-definable.

$(\Leftarrow)$ Assume that $W_Z$ is $k$-definable. By Proposition \ref{prop-definable}, there is $M\in \GL_m(k)$ such that
$$
    M(S^{\leq d}(\bfv_1), \cdots, S^{\leq d}(\bfv_m))=\begin{pmatrix} \tilde{W} \\ 0 \end{pmatrix}
$$
where $\tilde{W}$ is an $m\times m$ matrix. Denote the $i$th-row of $M$ by $\bfm_i$ for $i=m+1,\cdots, {n^2+d \choose d}$. Let $\bfx=(1,x_{1,1},\cdots, x_{1,1}^{\mu_{1,1}}\cdots x_{n,n}^{\mu_{n,n}},\cdots)$ where $\mu_{1,1}+\mu_{1,2}+\cdots+\mu_{n,n}\leq d$. Then $P_i(X)=\bfm_i\cdot \bfx$ is a polynomial of degree at most $d$ with the coefficients in $k$,
where $i=m+1,\cdots, {n^2+d \choose d}$. One can easily verify that
 $$ Z\subseteq \Zero\left(P_m(X), \cdots, P_{n^2+d \choose d}(X)\right)\bigcap V_{inv}^n.$$
  Since $Z\in N_d(V^n_{inv})$, there are $Q_1(X), \cdots, Q_l(X)$ with degree at most $d$, which define $Z$. By the dimension argument, for each $i=1,\cdots,l$, $Q_i(X)$ is a $\bar{k}$-linear combinations of $P_m(X), \cdots, P_{n^2+d \choose d}(X)$. Therefore the above inclusion relation is actually an equality.
\end{proof}
The above two propositions indicate the following corollary.
\begin{cor}
\label{cor-bound-k-definable}
There is an integer $\ell$ such that for every $k$-definable element $Z$ of $N_d(V^n_{inv})$, the coefficients of the defining equations of $Z$ can be chosen to be rational functions whose degrees are not greater than $\ell$.
\end{cor}
\begin{proof}
   Let $Z$ be a $k$-definable element of $N_d(V^n_{inv})$. Then $W_Z$ is a $k$-definable subspace of $\ssV$ by Proposition \ref{prop-definable-corresponding}. Moreover, from the proof of Proposition \ref{prop-definable-corresponding}, the coefficients of the defining equations of $Z$ can be chosen to be the entries of the defining matrix of $W_Z$.  Hence to prove the corollary, it suffices to show that all $k$-definable subspaces of $\ssV$ are bounded by an integer $\ell$. This can be done by applying Proposition \ref{prop-bounded-definable-subspace} to $\ssV$.
\end{proof}

\section{Bounds for proto-Galois groups}


In this appendix, we shall find an integer $\tilde{d}$ depending on $n$ with the following property. For any algebraic subgroup $G$ of $\GL_n(C)$, there is an algebraic subgroup $H$ of $\GL_n(C)$, which is bounded by $\tilde{d}$, satisfying
$$
   (*): (H^\circ)^t\unlhd G^\circ \leq G \leq H.
$$
Most of results in this section appeared in the part~\Rmnum{3} of \cite{hrushovski}, where more families of algebraic subgroups of $\GL_n(C)$ that can be uniformly definable are given. Here we only present those we need. At the same tine, we will use the term ``bounded by $d$" instead of ``uniformly definable". As mentioned in Introduction, we elaborate the details of the proofs in \cite{hrushovski} and present the explicit estimates of the bounds.
Meanwhile, we will show how to compute a set of generators of the character group of a given connected algebraic subgroup. The following notation will be used frequently.

\begin{notation}
Let $H$ be an algebraic subgroup of $\GL_n(C)$ and $S$ an arbitrary subset of $H$.
\[
\begin{minipage}{15cm}
      $\calF$: a family of algebraic subgroups of $\GL_n(C)$;\\[1mm]
      $H_\calF$: the intersection of all $H'\in \calF$ with $H\subseteq H'$;\\[1mm]
      $N_H(H')$: the normalizer of $H'$ in $H$ where $H'$ is an algebraic subgroup of $H$;\\[1mm]
      $\calF_{mt}(H)$: the family of maximal tori of $H$;\\[1mm]
      $\calF_{imt}(H)$: the family of intersections of maximal tori of $H$;\\[1mm]
      $\calF_{up}$: the family of subgroups of $\GL_n(C)$ generated by unipotent elements;\\[1mm]
      $X(H)$: the group of characters of $H$;\\[1mm]
      $H^t$: the intersection of kernals of all characters of $H$.
\end{minipage}
\]
\end{notation}
An algebraic subgroup $H$ of $\GL_n(C)$ is said to be bounded by $d$ if there are polynomials $Q_1(X), \cdots, Q_m(X)$ in $C[x_{1,1},x_{1,2},\cdots, x_{n,n}]$ with degree not greater than $d$ such that
$$H=\Zero(Q_1(X), \cdots, Q_m(X))\bigcap \GL_n(C).$$
For a family of algebraic subgroups of $\GL_n(C)$, say $\calF$, we say $\calF$ is bounded by $d$, if every element of $\calF$ is bounded by $d$. For an ideal $I$ in $C[x_{1,1},x_{1,2},\cdots, x_{n,n}]$, $I$ is said to be bounded by $d$ if there exist generators of $I$ whose degrees are not greater than $d$. Throughout this appendix, unless otherwise specified, subgroups always mean algebraic subgroups.
\subsection{Preparation lemmas}

To achieve the integer $\tilde{d}$, we need the following degree bounds from computational algebraic geometry. More details on these degree bounds can be found in \cite{dube,seidenberg,vandendries-schmidt}.
  For the moment, we assume that $I$ is an ideal in $C[x_1,\cdots,x_n]$. Then we have
  \begin{prop}
  \label{prop-elimination}
     Suppose that $I$ is bounded by $d$. Then there is $\gamma(n,d)$ in $\bN$ such that $I \cap C[x_1,\cdots, x_i]$ is bounded by $\gamma(n,d)$.
  \end{prop}
  \begin{remark}
    \label{remark-bounds}
    By Groebner bases computation, $\gamma(n,d)$ can be chosen as $2\left(\frac{d^2}{2}+d\right)^{2^{n-1}}$, which is less than $(d+1)^{2^n}$. The reader is referred to \cite{dube} for more details. \end{remark}
The following several lemmas play the key role in this appendix.

  \begin{lemma}
  \label{lem-bound5}
     Let $H$ be a subgroup of $\GL_n(C)$ bounded by $d$. Then there exists a family $\calF_{ad}(H)$ of subgroups of $H$ bounded by $\max\{d,n\}$ such that for any connected subgroup $H'$ of $H$, $N_H(H')\in \calF_{ad}(H)$. Particularly, $\calF_{ad}(\GL_n(C))$ is bounded by $n$.
  \end{lemma}
  \begin{proof}
     Let $\gl(H)$ be the Lie algebra of $H$. Consider the adjoint action of $H$ on $\gl(H)$. Then $\gl(H')$ is a subspace of $\gl(H)$ and $N_H(H')$ is the stabilizer of $\gl(H')$ under the adjoint action.
     Let $B_1,\cdots, B_l$ be a basis of $\gl(H')$ and the $\bfe_{i,j}$ a basis of $\Mat_n(C)$. Then for any $h\in N_{H}(H')$, there is $g_h\in \GL_{n^2}(C)$ such that
     $ h(\bfe_{1,1},\cdots, \bfe_{n,n})h^{-1}=(\bfe_{1,1},\cdots,\bfe_{n,n})g_h$. It is easy to see that the entries of $g_h$ are of the form $P_{l,m}(h)/\det(h)$ where the $P_{l,m}(X)$ are polynomials with degree at most $n$. Assume that $B_s=(\bfe_{1,1},\cdots,\bfe_{n,n})\bfb_s$ for $s=1,\cdots,l$ where $\bfb_s=(b_{s,i,j}) \in C^{n^2}$. Then since $hB_s h^{-1}\in \gl(H')$, there are $a_{s, 1},\cdots,a_{s, l}\in C$ such that $hB_s h^{-1}=\sum_\xi a_{s, \xi} B_\xi$. In other words,
     $$
         hB_s h^{-1}=\sum_{i,j}b_{s, i,j} h\bfe_{i,j}h^{-1}=(\bfe_{1,1},\cdots,\bfe_{n,n})g_h\bfb_s=(\bfe_{1,1},\cdots,\bfe_{n,n})
         \sum_{\xi=1}^l a_{s, \xi} \bfb_\xi.
     $$
     That is
     \begin{equation*}
     \label{eqn-nonhomogeneous}
         g_h\bfb_s=\sum_{\xi=1}^l a_{s, \xi} \bfb_\xi.
     \end{equation*}
     The above nonhomogeneous linear equations has solutions if and only if
     $$
       \rank(\bfb_1,\cdots,\bfb_l)=\rank(\bfb_1,\cdots,\bfb_l,g_h\bfb_s).
     $$
     This leads to the equations that together with the defining equations of $H$ define $N_H(H')$. Since the entries of $g_h$ are of the form $P_{l,m}(h)/\det(h)$ where the $P_{l,m}(X)$ are polynomials with degree at most $n$, the defining ideal of $N_H(H')$ is generated by those of $H$ and the polynomials with degrees $\leq n$. Hence $\calF_{ad}(H)$ is bounded by $\max\{d,n\}$. In particular, when $H=\GL_n(C)$, $\calF_{ad}(\GL_n(C))$ is bounded by $n$.
  \end{proof}

  Let $\{\tau_{H,\lambda}: H \rightarrow \GL_\mu(C)| H\in \calF, \lambda\in \Lambda\}$ be a family of morphisms from elements of $\calF$ to $\GL_{\mu}(C)$ where $\mu$ is a positive integer and $\Lambda$ is a set. Assume that $\tau_{H,\lambda}=(P_{i,j}^{H,\lambda}(X)/Q^{H,\lambda}(X))$. We will say that the $\{\tau_{H,\lambda}\}$ are bounded by $m$ if $\deg(P^{H,\lambda}_{i,j}(X))\leq m$ and $\deg(Q^{H,\lambda}(X)) \leq m$.
  \begin{lemma}
    \label{lem-morphisms}
    Let $\{\tau_{H,\lambda}\,\,|\,\,H\in \calF, \lambda\in \Lambda\}$ be as above. Assume that $\calF$ is bounded by $d$ and $\{\tau_{H,\lambda}\,\,|\,\,H\in \calF, \lambda\in \Lambda\}$ is bounded by $m$. Then
    \begin{itemize}
       \item [$(a)$] $\{\tau_{H,\lambda}(H)\}$ is bounded by $(\bar{d}+1)^{2^{\mu^2+n^2}}$
     where $\bar{d}=\max\{m+1,d\}$.
       \item [$(b)$] if $\calF'$ is a family of subgroups of $\GL_{\mu}(C)$ bounded by $d'$, then $$\{\tau^{-1}_{H,\lambda}\left(H'\cap \tau_{H,\lambda}(H)\right)| H'\in \calF', H\in\calF, \lambda\in \Lambda\}$$ is bounded by $\max\{d, md'\}.$
    \end{itemize}
  \end{lemma}
  \begin{proof}
     Assume that $H$ is defined by $S_H$, a set of polynomials with degree $\leq d$.

     $(a)$ $\tau_{H,\lambda}(H)$ is defined by
     $$
       \left\langle Q^{H,\lambda}(X)y_{1,1}-P^{H,\lambda}_{1,1}(X), \cdots, Q^{H,\lambda}(X)y_{\mu,\mu}-P^{H,\lambda}_{\mu,\mu}(X), S_H\right\rangle\bigcap C[y_{1,1},y_{1,2},\cdots, y_{\mu,\mu}].
     $$
     By Proposition \ref{prop-elimination}, $\tau_{H,\lambda}(H)$ is bounded by $(\bar{d}+1)^{2^{\mu^2+n^2}}$
     where $\bar{d}=\max\{m+1,d\}$.

     $(b)$ Assume that $H'$ is defined by $g_1(Y),\cdots, g_s(Y)$ where $\deg(g_i(Y))\leq d'$. Then one can see that $\tau^{-1}_{H,\lambda}\left(H'\cap \tau_{H,\lambda}(H)\right)$ is defined by
     $$
        S_H, \,\,g_1\left(\left(\frac{P^{H,\lambda}_{i,j}(X)}{Q^{H,\lambda}(X)}\right)\right), \cdots, g_s\left(\left(\frac{P^{H,\lambda}_{i,j}(X)}{Q^{H,\lambda}(X)}\right)\right).
     $$
     Clearing the denominators, we obtain the defining polynomials of $\tau^{-1}_{H,\lambda}\left(H'\cap \tau_{H,\lambda}(H)\right)$ in $C[x_{1,1},x_{1,2},\cdots, x_{n,n}]$, whose degrees are not greater than $\max\{d, md'\}$.
  \end{proof}
  Given a non-negative integer $d$, set
  $$
     d^*=\max_i \left\{{{n^2+d \choose d}\choose i}^2\right\},\,\,n^*=d^* d {n^2+d \choose d}.
  $$
  \begin{prop}
  \label{prop-embedding}
     Assume that $\calF$ is bounded by $d$. Then for any subgroup $H'\subseteq \GL_n(C)$ and $H\in \calF$ with
     $H \unlhd H'$, there is a family of morphisms $\tau_{H',H}: H'\rightarrow \GL_{d^*}(C)$ with $\ker(\tau_{H',H})=H$, which are bounded by $n^*$.  Furthermore, if $H'$ varies among a family of subgroups of $\GL_n(C)$ bounded by $d'$, then $\{\tau_{H',H}(H')|H'\in \calF', H\in \calF \,\,\mbox{with}\,\, H\unlhd H'\}$ is bounded by
     $$(\bar{d}+1)^{2^{(d^*)^2+n^2}},\,\,
     \mbox{where} \,\,\bar{d}=\max\{n^*+1,d'\}.$$
  \end{prop}
  \begin{proof}
     $C[x_{1,1},x_{1,2},\cdots, x_{n,n}]_{\leq d}$ is a $C$-vector space with dimension $n^2+d \choose d$. The group $\GL_n(C)$ acts naturally on $C[x_{1,1},x_{1,2},\cdots, x_{n,n}]_{\leq d}$, which is defined as follows
     $$
        \forall\,\, g\in \GL_n(C),\,\, P(x)\in C[x_{1,1},x_{1,2},\cdots, x_{n,n}]_{\leq d}, \,\, g\cdot P(X)=P(Xg).
     $$
     Suppose that $H\in \calF$. Let
     $$
        I_{\leq d}(H)=\{P(X)\in C[x_{1,1},x_{1,2},\cdots, x_{n,n}]_{\leq d} \,\,| \,\,P(H)=0\}.
     $$
     It is also a $C$-vector space of finite dimension. Since $I_{\leq d}(H)$ defines $H$, $H=\stab(I_{\leq d}(H))$. Let $\nu=\dim_C(I_{\leq d}(H))$ and
     $$
         E=\bigwedge^\nu C[x_{1,1},x_{1,2},\cdots, x_{n,n}]_{\leq d}.
     $$
     Then $$\dim_C(E)={{n^2+d \choose d} \choose \nu} \,\,\mbox{and}\,\, \bigwedge^\nu I_{\leq d}(H)=C\bfv\,\,\mbox{for some $\bfv\in E$}.$$ The action of $\GL_n(C)$ on $C[x_{1,1},x_{1,2},\cdots, x_{n,n}]_{\leq d}$ induces an action of $\GL_n(C)$ on $E$. We will still use $\cdot$ to denote this action. It is easy to see that $H=\stab(C\bfv)$. Let $U=\oplus U_{\chi}$ where the $\chi$ are characters of $H$ and $$U_{\chi}=\{\bfu\in E \,\,|\,\, h\cdot\bfu=\chi(h)\bfu \}.$$
     Note that the above direct sum runs over a finite set. Assume that $U=\oplus_{i=1}^s U_{\chi_i}$.
     It is clear that $\bfv\in U_{\chi_i}$ for some $i$. Let $H'$ be a subgroup of $\GL_n(C)$ satisfying that $H\unlhd H'$. Then $U$ is invariant under the action of $H'$.  Let $\calL$ be the set of $C$-linear maps from $U$ to $U$ which leave all $U_{\chi_i}$ unchanged. Then since $\dim_C(U)\leq \dim_C(E)$,
     $$\dim_C(\calL) \leq (\dim_C(U))^2 \leq (\dim_C(E))^2={{n^2+d \choose d}\choose \nu}^2.$$
      Let $\bfu_1,\cdots, \bfu_l$ be a suitable basis of $U$ such that under this basis, each element of $\calL$ is represented as the matrix $\diag(M_1,\cdots,M_s)$ where $M_i\in \Mat_{\dim(U_{\chi_i})}(C)$. Furthermore every matrix of the form $\diag(M_1,\cdots, M_s)$ where
      $M_i\in \Mat_{\dim(U_{\chi_i})}(C)$ represents an element of $\calL$. For any $h'\in H'$, there is $[h']\in \GL_l(C)$ such that $$(h'\cdot\bfu_1,\cdots,h'\cdot\bfu_l)=(\bfu_1,\cdots,\bfu_l)[h'].$$ By an easy calculation, the entries of $[h']$ are polynomials in those of $h'$ with degree $\leq d\nu$. For any $L\in \calL$, we will use $L^{\bfu}$ to denote the matrix in $\GL_l(C)$ satisfies that
      $$
         L((\bfu_1,\bfu_2,\cdots, \bfu_l))=(\bfu_1,\bfu_2,\cdots, \bfu_l)L^\bfu.
      $$
      The action of $H'$ on $U$ derives an adjoint action of $H'$ on $\calL$ as follows: for any $L\in \calL, h'\in H'$,
      \begin{align*}
      (h'\cdot L) ((\bfu_1,\cdots,\bfu_l))&=h'\cdot L(h'^{-1}\cdot \bfu_1,\cdots,h'^{-1}\cdot \bfu_l)=h'\cdot L((\bfu_1,\cdots,\bfu_l)[h']^{-1})\\
      &=h'\cdot((\bfu_1,\cdots,\bfu_l)L^\bfu [h']^{-1})=(\bfu_1,\cdots,\bfu_l)[h']L^\bfu [h']^{-1}.
      \end{align*}
      Fix a basis of $\calL$, say $L_1,\cdots, L_m$, where $m\leq {{n^2+d \choose d} \choose \nu}^2$. Then the adjoint action induces a morphism from $H'$ to $\GL_m(C)$
     \begin{equation}
     \label{define-morphism}
          \tau_{H',H}:  H' \longrightarrow \GL_m(C), \,\,\tau_{H',H}(h')=\eta_{h'}
     \end{equation}
     where $\eta_{h'}\in \GL_m(C)$ satisfies that
     $$([h']L_1^\bfu [h']^{-1},\cdots,[h']L_m^\bfu [h']^{-1})=(L_1^\bfu,\cdots,L_m^\bfu)\eta_{h'}.$$ We will show that $\ker(\tau_{H',H})=H$.
     Suppose that $h'\in \ker(\tau_{H',H})$. Then $\eta_{h'}=I_m$. In other words, $[h']L^\bfu =L^\bfu [h']$ for all $L\in\calL$. It implies that
     $[h']$ is of the following form:
     $$[h']=\diag(\underbrace{c_1,\cdots, c_1}_{\dim(V_{\chi_1})},\cdots,\underbrace{c_s,\cdots,c_s}_{\dim(V_{\chi_s})}).$$
     Particularly, $h'\cdot\bfv=c_i \bfv$ for some $i$. Hence $h'\in H=\stab(C\bfv)$. One can easily see that $H\subseteq \ker(\tau_{H',H})$. Hence $H=\ker(\tau_{H',H})$.
     Since $m\leq d^*$, $\GL_m(C)$ can be naturally embedded into $\GL_{d^*}(C)$. Composing this embedding map with $\tau_{H',H}$ induces a morphism from $H'$ to $\GL_{d^*}(C)$ with kernel $H$. We will still denote this morphism by $\tau_{H',H}$. An easy calculation yields that
     $$
        \tau_{H',H}(X)=\left(\frac{P_{i,j}^{H',H}(X)}{Q^{H',H}(X)}\right)
     $$
     where $P_{i,j}^{H',H}(X), Q^{H',H}(X)$ are polynomials in $x_{i,j}$ and $Q^{H',H}(h')=\det([h'])$. Furthermore,
      the $P_{i,j}^{H',H}(h')$ are polynomials in the entries of $[h']$ with degree $\leq l$. Since the entries of $[h']$ are polynomials in those of $h'$ with degree $\leq d\nu$, the $P_{i,j}^{H',H}(h')$ and $Q^{H',H}(h')$ are polynomials in the entries of $h'$ with degree $\leq l d\nu$, which is not greater than
      $n^*$. This proves that $\{\tau_{H,\lambda}\,\,|\,\,H\in \calF, \lambda\in \Lambda\}$ is bounded by $n^*$.

     Finally, by Lemma \ref{lem-morphisms}, $\{\tau_{H',H}(H')|H'\in \calF', H\in \calF \,\,\mbox{with}\,\, H\unlhd H'\}$ is bounded by $(\bar{d}+1)^{2^{(d^*)^2+n^2}},\,\,
     \mbox{where} \,\,\bar{d}=\max\{n^*+1,d'\}$.
  \end{proof}
  \begin{remark}
     As linear algebraic groups, $\tau_{H',H}(H')$ is isomorphic to $H'/H$. Therefore Proposition \ref{prop-embedding} says that $H'/H$ can be uniformly embedded into $\GL_{d^*}(C)$ and $H'/H$ varies among a bounded family if $H'$ does.
  \end{remark}


  \begin{lemma}
  \label{lem-unipotent}
    $\calF_{up}$ is bounded by  $(2n^3+1)^{8^{n^2}}$.
  \end{lemma}
  \begin{proof}
    Assume that $H$ is a subgroup generated by unipotent elements. Then by (p.55, Proposition, \cite{humphreys} ), it is the product of at most $2\dim(H)$ one-dimensional unipotent subgroups.
    From (p. 96, Lemma C, \cite{humphreys}), we know that one-dimension unipotent subgroup is of the form:
    $$
       I_n+\bfm x+\frac{\bfm^2}{2}x^2+\cdots+\frac{\bfm^{n-1}}{(n-1)!}x^{n-1},\,\, \bfm^n=0,\,\, x\in C,
   $$
   where $\bfm\in \Mat_n(C)$ and $\bfm^n=0$. Hence $H$ has a polynomial parametrized representation
   $$
      (Y)=\prod_{i=1}^{2\dim(H)}\left(I_n+\bfm_i x_i+\frac{\bfm_i^2}{2}x_i^2+\cdots+\frac{\bfm_i^{n-1}}{(n-1)!}x_i^{n-1}\right).
   $$
   Note that $\dim(H)\leq n^2$.
   Then the polynomials in the above system contains at most $3n^2$ variables and are of degree not greater than $2n^3$. Eliminating all $x_i$ in the above system, we obtain the defining ideal of $H$. Then Proposition \ref{prop-elimination} yields the desired bound.
  \end{proof}

  \begin{lemma}
  \label{lem-maxitori}
   Both $\calF_{mt}(\GL_n(C))$ and $\calF_{imt}(\GL_n(C))$ are bounded by 1.
  \end{lemma}
  \begin{proof}
     Every maximal torus of $\GL_n(C)$ is conjugate to $(C^*)^n$. Hence it is equal to the intersection of $\GL_n(C)$ and a linear subspace of $\Mat_n(C)$. Consequently, $\calF_{mt}(\GL_n(C))$ is bounded by 1. As the intersection of linear subspaces of $\Mat_n(C)$ is still linear, any element of $\calF_{imt}(\GL_n(C))$ is the intersection of $\GL_n(C)$ and a linear subspace of $\Mat_n(C)$. So $\calF_{imt}(\GL_n(C))$ is also bounded by 1.
  \end{proof}
  \begin{lemma}
  \label{lem-kernels}
     Assume that $H$ is a connected subgroup of $\GL_n(C)$. Then $H^t$ is generated by all unipotent elements of $H$.
  \end{lemma}
  \begin{proof}
      Since $H/H^t$ is a torus, there are no nontrivial unipotent elements in $H/H^t$. Hence all unipotent elements of $H$ are in $H^t$.
      From  Lemma 2.1 in \cite{singer}, $H^t$ is generated by $(P, P)$ and $R_u(H)$ where $P$ is a Levi factor of $H$ and $R_u(H)$ is the unipotent radical of $H$. Moreover $(P,P)$ is semi-simple, so it is generated by unipotent elements. Therefore $H^t$ is generated by all unipotent elements of $H$.
  \end{proof}
\subsection{Main results}
  Let $J(n)$ be a Jordan bound, so that every finite subgroup of $\GL_n(C)$ contains a normal abelian subgroup of index at most $J(n)$. In the following, we will show the main results in this appdendix.
   Denote
     \begin{equation}
     \label{notation-kappa12}
        \kappa_1=\max_i\left\{{{n^2+(2n^3+1)^{8^{n^2}} \choose n^2} \choose i}^2\right\}\quad \mbox{and}\quad
        \kappa_2=\kappa_1(2n^3+1)^{8^{n^2}}{n^2+(2n^3+1)^{8^{n^2}} \choose n^2}.
     \end{equation}
  \begin{prop}
  \label{prop-familygroups}
     There exists an integer $I(n)$ that is not greater than $J\left(\max_i\left\{{\kappa_1^2+1 \choose i}\right\}\right)$ and a family $\calF$ of subgroups of $\GL_n(C)$ bounded by
     \begin{equation}
        \label{eqn-bound1}
           \kappa_3\triangleq\kappa_2(\kappa_1^2+1)\max_i\left\{{\kappa_1^2+1 \choose i}\right\}
      \end{equation}
     with the following property. For every subgroup $H$ of $\GL_n(C)$, there is $H'\in \calF$ such that
      \begin{itemize}
       \item [$(a)$] $H^\circ \leq H'$.
       \item [$(b)$] $H$ normalizes $H'$; so $H' \unlhd HH' \leq \GL_n(C)$.
       \item [$(c)$] $[H:H\cap H']=[HH':H']\leq I(n)$.
       \item [$(d)$] Every unipotent element of $H'$ lies in $H^\circ$.
      \end{itemize}
  \end{prop}
  We will show the proposition by separating three cases.
  \begin{lemma}
    \label{lem-case-finite}
     Proposition \ref{prop-familygroups} is true for finite groups with $I(n)=J(n)$ and $\calF$ is bounded by 1.
  \end{lemma}
  \begin{proof}
     Assume that $H\subset \GL_n(C)$ is a finite group.
     Let $\bar{H}$ be a normal abelian subgroup of $H$ with $[H:\bar{H}]\leq J(n)$. As a finite abelian subgroup of $\GL_n(C)$ is diagonalizable, $\bar{H}$ is contained in some maximal tori of $\GL_n(C)$. Let $H'$ be the intersection of maximal tori containing $\bar{H}$. Then $H'\in \calF_{imt}(\GL_n(C))$. Clearly, $H$ normalizes $H'$. Since $\bar{H}\subseteq H\cap H'$, $[H:H'\cap H]\leq [H:\bar{H}]\leq J(n)$. The only unipotent element of $H'$ is the identity. So $(a)-(d)$ hold for $H,H'$. The lemma follows from the fact that $\calF_{imt}(\GL_n(C))$ is bounded by $1$.
  \end{proof}
  \begin{lemma}
   \label{lem-case-torus}
   Assume that $H$ is a subgroup whose identity component is a torus. Then
      Proposition \ref{prop-familygroups} is true for $H$ with $I(n)=J\left(\max_i\left\{{n^2+1 \choose i}^2\right\}\right)$ and $\calF$ is bounded by
      $$(n^2+1)\max_i\left\{{n^2+1 \choose i}^2\right\}.$$
  \end{lemma}
  \begin{proof}
      Let $M=(H^\circ)_{\calF_{imt}(\GL_n(C))}$ and  $N=N_{\GL_n(C)}(M)$. It is easy to verify that $H$ normalizes $M$ and thus $H\subseteq N$. Since $M$ lies in the family $\calF_{imt}(\GL_n(C))$
      bounded by $1$, Lemma \ref{lem-bound5} implies that $N$ lies in a family $\calF_{ad}(\GL_n(C))$ bounded by $n$. Let $\tilde{n}=\max_i\{{n^2+1\choose i}^2\}$. By Proposition \ref{prop-embedding}, there is a morphism
     $$
         \tau_{N,M}: N \longrightarrow \GL_{\tilde{n}}(C)
     $$
     satisfies that $\ker(\tau_{N,M})=M$ and $\tau_{N,M}$ is bounded by $\tilde{n}(n^2+1)$.
     As $H^\circ \subseteq M$, $\tau_{N,M}(H)$ is a finite subgroup of $\GL_{\tilde{n}}(C)$. From Lemma \ref{lem-case-finite}, there is $\tilde{M}\in \calF_{imt}(\GL_{\tilde{n}}(C))$ such that $(a)$-$(c)$ hold for $\tau_{N,M}(H), \tilde{M}$ (with $I(\tilde{n})=J(\tilde{n})$). Let $H'=\tau_{N,M}^{-1}(\tilde{M}\cap \tau_{N,M}(N))$. Note that $\calF_{imt}(\GL_{\tilde{n}}(C))$ is bounded by 1. By Lemma \ref{lem-morphisms}, $H'$ is bounded by $\tilde{n}(n^2+1)$. We will show that the $H'$ satisfy $(a)$-$(d)$ with $I(n)=J(\tilde{n})$. It is clear that $H^\circ \leq H'$. For any $h\in H$, since $T_{M,N}(H)$ normalizes $\tilde{M}$ and $\tau_{N,M}(N)$,
     $$
        \tau_{N,M}(hH'h^{-1})=\tau_{N,M}(h)(\tilde{M}\cap \tau_{N,M}(N))\tau_{N,M}(h)^{-1}=\tilde{M}\cap \tau_{N,M}(N)=\tau_{N,M}(H').
     $$
     Therefore $hH'h^{-1} \subseteq H'$. This indicate that $hH'h^{-1}=H'$ for any $h\in H$. In other words, $H$ normalizes $H'$. This proves $(b)$.
     Since both $HH'$ and $H'$ contain $M$,
     \begin{align*}
        [HH':H']&=[\tau_{N,M}(HH'):\tau_{N,M}(H')]=[\tau_{N,M}(H)\tau_{N,M}(H'):\tau_{N,M}(H')]\\
        &=[\tau_{N,M}(H):\tau_{N,M}(H)\cap\tau_{N,M}(H')]=[\tau_{N,M}(H):\tilde{M}\cap\tau_{N,M}(H)]\leq J(\tilde{n}).
     \end{align*}
     This proves $(c)$. Suppose that $h'$ is an unipotent element of $H'$. Then $\tau_{N,M}(h')$ is an unipotent element of $\tilde{M}$. However $\tilde{M}$ consists of semi-simple elements. Hence $\tau_{N,M}(h')=1$. Then $h'\in M$. But $M$ is contained in a torus, so $h'=1$. This proves $(d)$.
     \end{proof}

     Now we are ready to prove Proposition \ref{prop-familygroups} for the general case.
     \begin{proof}
     Let $U=(H^\circ)^t$. By Lemma \ref{lem-kernels}, $U$ is generated by unipotent elements. Then it follows from Lemma \ref{lem-unipotent} that $U$ is bounded by $(2n^3+1)^{8^{n^2}}$. Let $N=N_{\GL_n(C)}(U)$. Lemma \ref{lem-bound5} indicates that $N$ lies in $\calF_{ad}(\GL_n(C))$ that is bounded by $n$. Let $\kappa_1$ and $\kappa_2$ be as in (\ref{notation-kappa12}).
     Using Proposition \ref{prop-embedding} again, there is a morphism
     $$
         \phi_{N,U}: N \longrightarrow \GL_{\kappa_1}(C)
     $$
     such that $\ker(\phi_{N,U})=U$ and $\phi_{N,U}$ is bounded by $\kappa_2$.
     We first prove that $H \leq N$. For any $h\in H$ and any character $\chi$ of $H^\circ$, $\chi(hXh^{-1})$ is a character of $H^\circ$. Hence for any $u\in U$, $\chi(huh^{-1})=1$. So $huh^{-1}\in U$ for any $u\in U$. In other words, $H$ normalizes $U$. So $H\leq N$. As $H^\circ/U$ is a torus, $\phi_{N,U}(H)^\circ$ is a torus in $\GL_{\kappa_1}(C)$. Lemma \ref{lem-case-torus} implies that there is $M' \leq \GL_{\kappa_1}(C)$ bounded by $(\kappa_1^2+1)\max_i\left\{{\kappa_1^2+1 \choose i}^2\right\}$ such that $(a)$-$(d)$ hold for $\phi_{N,U}(H), M'$ with $I(n)=J\left(\max_i\left\{{\kappa_1^2+1 \choose i}\right\}\right)$. Let $H'=\phi^{-1}_U(M'\cap \phi_{N,U}(N))$. Then by Lemma \ref{lem-morphisms}, $H'$ is bounded by $\kappa_3$, where $\kappa_3$ is defined in (\ref{eqn-bound1}).
      The similar arguments as in the proof of Lemma \ref{lem-case-torus} implies $(a)$-$(c)$ hold for $H,H'$ with $I(n)=J\left(\max_i\left\{{\kappa_1^2+1 \choose i}\right\}\right)$. Now let us show that $(d)$ holds. Assume that $u$ is an unipotent element of $H'$. Then $\phi_{N,U}(u)$ is an unipotent element of $M'$. Since $M'$ and $\phi_{N,U}(H)^\circ$ satisfy $(d)$, $\phi_{N,U}(u)\in \phi_{N,U}(H)^\circ$. Whereas $\phi_{N,U}(H)^\circ$ is a torus, $\phi_{N,U}(u)=1$. Thus $u\in U\subseteq H^\circ$.
  \end{proof}
  \begin{prop}
  \label{prop-boundedfamily}
     Let $I(n)$ and $\kappa_3$ be as in Proposition \ref{prop-familygroups}.
     Then there exists a family $\tilde{\calF}$ of subgroups of $\GL_n(C)$ bounded by
     $$
        \tilde{d}\triangleq (\kappa_3)^{I(n)-1}
     $$
      with the following property. For any subgroup $H$ of $\GL_n(C)$, there exists $\tilde{H}\in \tilde{\calF}$ such that $H \leq \tilde{H}$, and every unipotent element of $\tilde{H}$ lies in $H^\circ$.
  \end{prop}
  \begin{proof}
      Let $\calF$ be as in Proposition \ref{prop-familygroups} and
      $$
         \tilde{\calF}=\{\bar{H}\,\, |\,\, \exists\,\, M\in \calF, M \unlhd \bar{H}, [\bar{H}:M]\leq I(n) \}.
      $$
       Every element of $\tilde{\calF}$ is the union of at most $I(n)$ cosets of some element in $\calF$. It is well-known that the union of two varieties is defined by the product of their corresponding defining polynomials.
      Hence $\tilde{\calF}$ is bounded by $\tilde{d}$.
      Assume that $H$ is a subgroup of $\GL_n(C)$. Let $H'$ be an element in $\calF$ such that $(a)$-$(d)$ in Proposition \ref{prop-familygroups} hold for $H,H'$. Let $\tilde{H}=HH'$. Then $\tilde{H}\in \tilde{\calF}$ by Proposition~\ref{prop-familygroups} (c). The unipotent elements of $\tilde{H}$ lie in $\tilde{H}^\circ$ and then in $(H')^\circ$. As the unipotent elements of $H'$ lie in $H^\circ$, so do the unipotent elements of $\tilde{H}$.
  \end{proof}
  \begin{cor}
  \label{cor-boundeddegree}
     Let $\tilde{\calF}$ be the family as in Proposition \ref{prop-boundedfamily}. Then for any subgroup $H$ of $\GL_n(k)$, there is $\tilde{H}\in \tilde{\calF}$ such that $$ (\tilde{H}^\circ)^t \unlhd H^\circ \leq H \leq \tilde{H}.$$
  \end{cor}
  \begin{proof}
     By Proposition \ref{prop-boundedfamily}, there is $\tilde{H}\in \tilde{\calF}$ such that $H\leq \tilde{H}$ and the unipotent elements of $\tilde{H}$ lie in $H^\circ$. Then the corollary follows from Lemma \ref{lem-kernels} and the fact that $(\tilde{H}^\circ)^t$ is normal in $\tilde{H}^\circ$.
  \end{proof}
  In the following, $H$ is assumed to be connected.
  Proposition \ref{prop-embedding} allows us to bound the degrees of generators of $X(H)$. $X(H)$ can be viewed as a subset of $C[H]$, the coordinate ring of $H$. The morphism $\varphi: H \rightarrow H'$ induces a group homomorphism $\varphi^\circ: X(H')\rightarrow X(H)$.
  \begin{prop}
  \label{prop-boundforcharacters}
   Let $\kappa_2$ be as in (\ref{notation-kappa12}). Then there are generators of $X(H)$, which are represented by polynomials bounded by $\kappa_2$.
  \end{prop}
  \begin{proof}
     From Lemmas \ref{lem-kernels} and \ref{lem-unipotent}, $H^t$ is bounded by $(2n^3+1)^{8^{n^2}}$. By Proposition \ref{prop-embedding}, there is a morphism $\tau_{H,H^t}: H \rightarrow \GL_{\kappa_1}(C)$ satisfying that
  $\ker(\tau_{H,H^t})=H^t$ and $\tau_{H,H^t}$ is bounded by $\kappa_2$, where $\kappa_1$ is as in (\ref{notation-kappa12}). $\tau_{H,H^t}$ is of the from
  $$
     \left(\frac{P^{H,H^t}_{i,j}(X)}{Q^{H,H^t}(X)}\right)\quad \mbox{where}\,\,\deg(P^{H,H^t}_{i,j}(X)) \leq \kappa_2,\,\, \deg(Q^{H,H^t}(X))\leq \kappa_2.
  $$
  From the proof of Proposition \ref{prop-embedding}, $Q^{H,H^t}(I_n)=1$ and for any $h,h'\in H$,
  $$Q^{H,H^t}(hh')=\det([hh'])=\det([h][h'])=\det([h])\det([h'])=Q^{H,H^t}(h)Q^{H,H^t}(h').$$
  It implies that $Q^{H,H^t}(X)\in X(H)$. Notice that $X((C^*)^{\kappa_1})$ is generated by the characters $y_1,\cdots, y_{\kappa_1}$ and so is the group of characters of any its subgroup.
  Since $\tau_{H,H^t}(H)$ is a torus in $\GL_{\kappa_1}(C)$, it is conjugate to a subgroup of $(C^*)^{\kappa_1}$. So $X(\tau_{H,H^t}(H))$ is generated by some linear polynomials. $\tau_{H,H^t}$ induces a group homomorphism:
  \begin{align*}
     \tau_{H,H^t}^\circ: X(\tau_{H,H^t}(H))&\rightarrow X(H)\\
                             \chi'&\rightarrow \chi'\circ \tau_{H,H^t}
  \end{align*}
  For any $\chi \in X(H)$ and $h\in H$, since $\chi(hh')=\chi(h)$ for all $h'\in H^t$, there is $g\in C[\tau_{H,H^t}(H)]$ such that
  $g\circ\tau_{H,H^t}=\chi$. One can verify that $g$ is actually a character of $\tau_{H,H^t}(H)$. Therefore $\tau_{H,H^t}^\circ$ is surjective. Let $L_1,\cdots, L_s$ be linear polynomials, which generate $X(\tau_{H,H^t}(H))$. Then $L_1\circ \tau_{H,H^t}, \cdots, L_s\circ \tau_{H,H^t} $ generate $X(H)$. Since $Q^{H,H^t}(X) \in X(H)$,
  $$Q^{H,H^t}(X), (L_1\circ\tau_{H,H^t})Q^{H,H^t}(X), \cdots, (L_s\circ\tau_{H,H^t})Q^{H,H^t}(X)$$
  still generate $X(H)$ and are polynomials bounded by $\kappa_2$.
  \end{proof}
 Let $P_1(X),\cdots, P_l(X)$ be in $C[x_{1,1}, x_{1,2}, \cdots, x_{n,n}]_{\leq \kappa_2}$ such that their images constitute a $C$-basis of  $C[x_{1,1}, x_{1,2}, \cdots, x_{n,n}]_{\leq \kappa_2}/(I(H))_{\leq \kappa_2}$, where $I(H)$ is the set of polynomials in $C[x_{1,1}, x_{1,2}, \cdots, x_{n,n}]$ that vanishes on $H$. Let $c_1,\cdots,c_l$ be indeterminates and $P_\bfc(X)=\sum_{i=1}^l c_i P_i(X)$, where
 $\bfc=(c_1,\cdots, c_l)$. Then the conditions
 $$P_\bfc (I_n)=1 \,\,\mbox{and}\,\,\forall \,\,h,h'\in H, \,\,P_\bfc(h)P_\bfc(h')-P_\bfc(hh')=0$$ induce a system of algebraic equations $J$ for $\bfc$. Moreover, we have the following proposition.
 \begin{prop}
    \label{prop-characters}
    $\dim(J)=0$ and for each zero $\bar{\bfc}$ of $J$ in $C^l$, $P_{\bar{\bfc}}(X)$ is a character of $H$.
 \end{prop}
 \begin{proof}
    Evidently, for each $\bar{\bfc}\in \Zero(J)\cap C^l$, $P_{\bar{\bfc}}(X)$ is a morphism from $H$ to $C^*$ and thus a character of $H$. Suppose that $\bar{\bfc},\bar{\bfc}'\in \Zero(J)\cap C^l$ and $P_{\bar{\bfc}}(h)=P_{\bar{\bfc}'}(h)$ for all $h\in H$. Then it implies that $P_{\bar{\bfc}}(X)-P_{\bar{\bfc}'}(X)\in (I(H))_{\leq \kappa_2}$. Hence
    $$
        \sum_{i=1}^l (\bar{c}_i-\bar{c}_i')P_i(X) \equiv 0 \mod (I(H))_{\kappa_2}
    $$
    where $\bar{\bfc}=(\bar{c_1},\cdots, \bar{c}_l)$ and $\bar{\bfc}'=(\bar{c_1}',\cdots, \bar{c}_l')$. Since $P_1(X), \cdots, P_l(X)$  modulo $(I(H))_{\kappa_2}$ are linearly independent over $C$.
    So $\bar{\bfc}=\bar{\bfc}'$. That is to say, the map $\varphi: \Zero(J)\cap C^l \rightarrow X(H)$ defined by $\varphi(\bar{\bfc})=P_{\bar{\bfc}}(X)$ is injective. Now assume that $\dim(J)>0$. Then $\Zero(J)\cap C^l$ is an uncountable set and so is $X(H)$. However, it is known that $X(H)$ is a countable set. This contradiction concludes the proposition.
 \end{proof}
 \begin{remark}
 \label{rem-characters}
    Given a connect subgroup $H$ of $\GL_n(C)$, Propositions \ref{prop-boundforcharacters} and \ref{prop-characters} allow us to compute a set of generators of $X(H)$ that are represented by polynomials in $C[x_{1,1},x_{1,2},\cdots, x_{n,n}]$.
 \end{remark}
 \end{appendix}

\end{document}